\documentclass[sumlimits,intlimits,namelimits%,draft
]{amsart}

\usepackage{bm}

\usepackage{amssymb,latexsym}
\usepackage[numeric,initials,bibtex-style,nobysame]{amsrefs}
\usepackage{amsthm}

\usepackage{a4wide}
\usepackage[american]{babel}

\usepackage{eucal}
\usepackage{mathrsfs}

\usepackage[dvipsnames]{xcolor}

\def\pg{\mathhexbox278}

\newcommand{\Cinf}{\ensuremath{\mathcal{C}^\infty}}

\newcommand{\D}{\ensuremath{\mathcal{D}}}
\newcommand{\G}{\ensuremath{\mathcal{G}}}
\renewcommand{\S}{\mathscr{S}}
\newcommand{\E}{\ensuremath{\mathscr{E}}}

\newcommand{\mb}[1]{\ensuremath{\mathbb{#1}}}
\newcommand{\N}{\mb{N}}

\newcommand{\R}{\mb{R}}
\newcommand{\C}{\mb{C}}

\newcommand{\vp}{\mathrm{vp}}
\newcommand{\sgn}{\mathop{\mathrm{sgn}}}

\renewcommand{\d}{\ensuremath{\partial}}
\newcommand{\diff}[1]{\frac{d}{d#1}}

\renewcommand{\div}{\mathop{\mathrm{div}}}

\newfont{\bl}{msbm10 scaled \magstep2}

\newtheorem{theorem}{Theorem}[section]
\newtheorem{lemma}[theorem]{Lemma}
\newtheorem{proposition}[theorem]{Proposition}

\newtheorem{corollary}[theorem]{Corollary}
\newtheorem{assumption}[theorem]{Assumption}
\theoremstyle{definition}
\newtheorem{remark}[theorem]{Remark}
\newtheorem{example}[theorem]{Example}

\newcommand{\beq}{\begin{equation}}
\newcommand{\eeq}{\end{equation}}
\newcommand{\col}{\colon}
\newcommand{\dis}[2]{\langle #1 , #2 \rangle}
\newcommand{\inp}[2]{\langle #1 | #2 \rangle}  %math mode
\newcommand{\notmid}{\mid\kern-0.5em\not\kern0.5em}

\newcommand{\norm}[2]{{\left\| #1 \right\|}_{#2}}

\newcommand{\al}{\alpha}
\newcommand{\be}{\beta}

\newcommand{\de}{\delta}
\newcommand{\eps}{\varepsilon}

\newcommand{\vphi}{\varphi}

\newcommand{\om}{\omega}

\newcommand{\supp}{\mathop{\mathrm{supp}}}

\newcommand{\ovl}[1]{\overline{#1}}

\newcommand{\Ra}{R_\alpha}
\newcommand{\Za}{\dot{Z}_\alpha}
\newcommand{\e}{\textup{e}}
\newcommand{\ud}{\underline{\d}}

\usepackage{tikz-cd}
\usepackage{tikz}

\usepackage{stackengine,scalerel}
\newcommand\dAlauxx{%
  \Shortstack{\rule{11pt}{.7pt}\\
    \rule{.7pt}{9pt}\kern9pt\rule{1.3pt}{9pt}\\
    \rule{11pt}{1.35pt}}%
}
\newcommand\dAl{%
  \setstackgap{S}{0pt}%
  \setstackEOL{\\}%
  \scalerel*{\kern1pt\dAlauxx\kern1pt}{\Delta}%
}
\newcommand{\dal}{\mathop{\dAl}}

\newcommand{\xip}{\widetilde{\xi}^{\,\,\raisebox{0.3ex}{*}}\!}

\newcommand{\Uel}{U_{\text{ele}}}
\newcommand{\Uma}{U_{\text{mag}}}

\begin{document}

\pagestyle{plain}

\title{Modeling of point charges with distributions, regularizations, and generalized functions}

\author{G\"unther H\"ormann}

\author{Nathalie Tassotti}

\address{Fakult\"at f\"ur Mathematik\\
Universit\"at Wien, Austria}

\email{guenther.hoermann@univie.ac.at}



\date{\today}

\begin{abstract} We apply regularization and generalized function techniques to the classical electromagnetic field of a point-charge, following a strategy by the physicist A.\ Gsponer (\cite{GEJPCorr}-\cite{GsponerCTED}), thereby extending it and adding rigor. We show how the Li\'{e}nard-Wiechert potential emerges essentially from the basic geometry of Minkowski space, namely by action of the d'Alembertian on a generating vector field, which is defined via a spacetime interval between an observer and the point-charge at retarded proper time.  
Furthermore, for a charged particle in its rest frame, we discuss generalized functions aspects of the electric monopole, magnetic dipole, electron singularity, and self-energy, where infinitely large generalized numbers occur whose concrete representations can be applied in field renormalization. 

\end{abstract}

\maketitle

%\maketitle
%
%\begin{abstract} 
%We discuss spectral properties of a regularization approach to a Schr\"odinger equation set-up for the diffraction of a quantum particle at almost planar patterns. Physically meaningful initial values and potentials are modeled in terms of regularizing families and the solutions can be interpreted as generalized functions. We establish spectral and scattering theoretical properties of the regularizing solution families and provide some comparison with the more direct approximations and simplifications used in physics.
%\end{abstract}
%
%
%\noindent
%{\bf Keywords:} Schr\"odinger equation, regularization methods, generalized functions.

%
%\subjclass[2020]{81Q05, 35P25, 46F30} 
%
%\keywords{Schr\"odinger equation, regularization methods, generalized functions.}
%
%\date{\today}
%
%

%%%%%%%%%%%%%%%%%%%%%%%%%%%%%%%%%%
%%%%%%%%%%%%%%%%%%%%%%%%%%%%%%%%%%
%%%%%%%%%%%%%%%%%%%%%%%%%%%%%%%%%%

%%%%%%%%%%%%%% internal Section 0 %%%%%%%%%
%\setcounter{section}{-1}

\section{Introduction}

The analytic investigation of electric point-charges in interaction with their own field is notorious for its difficulties with singularities and lead to artefacts like preacceleration and runaway solutions (cf.\ \cite[Sections 4.3 and 5.5]{Parrott} or \cite[Section 8.4]{Thirring:01}). One approach is to regularize the physical quantities while still keeping the geometric nature of point particles intact.  The field dynamics has been studied by regularizations in the context of Colombeau generalized functions in \cite{HK:98} and \cite{HK:00}. The basic set-up in the context of electrodynamics is also described in \cite{GsponerCTED}. Although some of the inconsistencies could be avoided without regularizations by employing specific initial or boundary conditions or axiomatically require a modified action principle (see \cite[Chapter 6]{Rohrlich}), such attempts cannot always be considered satisfactory (cf.\ \cite[Remark (8.4.21)]{Thirring:01})

The Colombeau-based regularization approach for the electromagnetic field of a point-charge has been substantially extended in \cite{GEW}, starting from pure Minkowski space notions of spatial distance and retarded proper time and constructing a generalized vector field $\Phi$, which we describe in Section \ref{SecMink}. 

The action of the d'Alembertian on $\Phi$ produces the Li\'{e}nard-Wiechert potential and an additional potential with relation to weak interaction. We give a refined and more detailed analysis of the core of this process in Section \ref{SecLWP}. In other words, $\Phi$ satisfies a vector wave equation with the Li\'{e}nard-Wiechert potential and a certain weak interaction term as a right-hand side.

In Section \ref{SecSelf}, we focus on a charged particle in its rest frame, where the Li\'{e}nard-Wiechert potential reproduces the Coulomb potential, along the lines of and extending \cite{GsponerJMP}, discussing the electric monopole, magnetic dipole, electron singularity, and self-energy. We show that for a certain class of regularizations both the electric as well as the magnetic self-energy are infinitely large generalized numbers whose representations are very concrete in terms of the underlying regularization. The latter may give some support in renormalization procedures.

The modeling of both the Li\'{e}nard-Wiechert potential and the Coulomb potential includes an essential factor in the form of a generalized function $\Upsilon$ as introduced and specified in \cites{GEW,GEJP,GPOT}. In Section \ref{HUps} we discuss in various approaches, by distribution theory, distributional products, or by nonlinear theories of generalized functions, why $\Upsilon$ could be considered equal or associated to the Heaviside function.

\subsection*{Generalized functions and coherence properties}

For the convenience of the reader, we review the very basics from the theory of Colombeau generalized functions. It can be considered an extension of the so-called \emph{sequential approach to distributions}, where each distribution is represented by approximating regularizing weakly convergent sequences modulo null sequences.

Colombeau regularization methods model nonsmooth objects by approximating nets of smooth functions, regardless of convergence, but with conditions of \emph{moderate} asymptotics, and identify regularizing nets whose differences compared to the moderateness scale are \emph{negligible}. A comprehensive modern introduction to Colombeau theory is \cite{GKOS:01}, although we will also make use of constructions and notations from \cite{Garetto:05b}, where  generalized functions based on a locally convex topological vector space $E$ are defined. Let the topology on $E$ be given by the family of seminorms $\{p_j\}_{j\in J}$, then the elements of  
$$
 \mathcal{M}_E := \{(u_\eps) \in E^{(0,1]}:\, \forall j\in J\,\, \exists N\in\N\quad p_j(u_\eps)=O(\eps^{-N})\, \text{as}\, \eps\to 0\}
$$
and
$$
 \mathcal{N}_E := \{(u_\eps) \in E^{(0,1]}:\, \forall j\in J\,\, \forall q\in\N\quad p_j(u_\eps)=O(\eps^{q})\, \text{as}\, \eps\to 0\},
$$ 
are called \emph{$E$-moderate} and \emph{$E$-negligible}, respectively. With respect to the componentwise operations, e.g., $(u_\eps) + (v_\eps) := (u_\eps + v_\eps)$ etc., $\mathcal{N}_E$ becomes a vector subspace of $\mathcal{M}_E$.  The \emph{Generalized functions based on $E$} are defined as the factor space $\G_E := \mathcal{M}_E / \mathcal{N}_E$. If, in addition, $E$ is an associative commutative algebra with continuous multiplication and possessing commuting continuous linear maps $\d_j \col E \to E$ ($j=1,\ldots,n$) that satisfy the Leibniz rule $\d_j( f \cdot g) = (\d_j f) \cdot g + f \cdot \d_j g$,  then $\mathcal{N}_E$ is an ideal in $\mathcal{M}_E$ and $\G_E$ is  a differential algebra as well.

Specific choices of the underlying space $E$ yield the standard Colombeau algebras of generalized functions. For example, equipping $E=\C$ with the absolute value gives the generalized complex numbers $\G_{\C} = \widetilde{\C}$.  If $\Omega \subseteq \R^d$ is open, then $E=\Cinf(\Omega)$ with the topology of compact uniform convergence of all derivatives provides the so-called special Colombeau algebra $\G_{\Cinf(\Omega)}=\G(\Omega)$. Recall that $\Omega \mapsto \G(\Omega)$ is a fine sheaf, thus, in particular, the restriction $u|_B$ of $u\in\G(\Omega)$ to an arbitrary open subset $B \subseteq \Omega$ is well-defined and yields $u|_B \in \G(B)$. Moreover, we may embed $\D'(\Omega)$ into $\G(\Omega)$ by appropriate localization and convolution regularization. 

If $E \subseteq \D'(\Omega)$, then certain generalized functions can be projected into the space of distributions by taking  weak limits: We say that $u \in \G_E$ is \emph{associated} with $w \in \D'(\Omega)$, if $u_\eps \to w$ in $\D'(\Omega)$ as $\eps \to 0$ holds for any (hence every) representative $(u_\eps)$ of $u$. This fact is also denoted by $u \approx w$. 

On an open strip of the form $\Omega_T = \R^n \times\, ]0,T[ \,\subseteq \R^{n+1}$ (with $T > 0$ arbitrary) we have the spaces $E = H^\infty({\Omega_T}) = \{ h \in \Cinf(\Omega_T) : \d^\al h \in L^2(\Omega_T) \; \forall \al\in \N^{n+1}\}$ with the family of (semi-)norms 
$$
  \norm{h}{H^k} = \Big( \sum_{|\al| \leq k} 
    \norm{\d^\al h}{L^2}^2\Big)^{1/2}
   \quad (k\in \N), 
$$
or also   
$E = W^{\infty,\infty}({\Omega_T}) = \{ h \in \Cinf(\Omega_T) : \d^\al h \in L^\infty(\Omega_T) \; \forall \al\in \N^{n+1}\}$  with the family of (semi-)norms 
$$
  \norm{h}{W^{k,\infty}} = \max_{|\al| \leq k} \norm{\d^\al h}{L^\infty} \quad (k\in \N). 
$$
Clearly, $\Omega_T$  satisfies the strong local Lipschitz property \cite[Chapter IV, 4.6, p.\ 66]{Adams:75}, hence every element of $H^\infty(\Omega_T)$ and $W^{\infty,\infty}(\Omega_T)$ belongs to $\Cinf(\ovl{\Omega_T})$ by the Sobolev embedding theorem \cite[Chapter V, Theorem 5.4, Part II, p.\ 98]{Adams:75}.

We will occasionally use the notation  $\G_\infty (\R^n \times [0,T]) := \G_{W^{\infty,\infty}({\Omega_T})}$.
A generalized function $u \in \G_{\infty}(\R^n \times [0,T])$ is thus represented by a net $(u_\eps)$ with $u_\eps \in W^{\infty,\infty}(\R^n \times\, ]0,T[)$ and  the moderateness property
$$
    \forall k \, \exists m: \quad 
    \norm{u_{\eps}}{W^{k,\infty}} = O(\eps^{-m}) \quad (\eps \to 0).
$$
If $(\widetilde{u_{\eps}})$ is another representative of $u$, then 
$$
    \forall k \, \forall p: \quad 
    \norm{u_{\eps} - \widetilde{u_{\eps}}}{W^{k,\infty}} = O(\eps^{p}) 
    \quad (\eps \to 0).
$$
Note that by Young's inequality (\cite[Proposition 8.9.(a)]{Folland:99}) any standard convolution regularization with a scaled mollifier of Schwartz class provides an embedding  $L^p \hookrightarrow \G_{\infty}$ ($1 \leq p \leq \infty$).  

\begin{remark}
Colombeau-generalized functions from $\G_{H^\infty} (\R^n \times [0,T])$ have been employed successfully in the context of quantum physics, namely as solutions to the Cauchy problem for the Schr\"odinger equation. On the one hand, results  on existence, uniqueness, and coherence---in the form of association with distributional solutions---have been established (cf.\ \cite{Hoermann:11}). On the other hand, also limiting behavior of particular solutions, asymptotics of scattering, and spectral properties have been studied in this context (cf.\  \cite{Hoermann:17, HOS:25}).
\end{remark}

%%%%%%%%%%%%%%%%%%%%%%%%%%%%%%%%%
\section{The basic quantities and regularizations on Minkowski space}\label{SecMink}

\subsection{Set-up and notation}

The underlying spacetime is Minkowski space $\R^4$ with the Lorentz metric $(g_{\al \be})_{0 \leq \al, \be \leq 3}$ of signature $(+,-,-,-)$. We will also employ abstract index notation and summation convention (cf.\ \cite[Chapter 2]{PR1}).

We consider the relativistic four-position $Z(\tau)$ of a moving charged point particle in spacetime,  parametrized by eigentime $\tau$. Thus $Z \col \R \to \R^4$, $\tau \mapsto Z(\tau)$, is a future-directed timelike smooth curve  with $\dot{Z_\mu}(\tau) \dot{Z^\mu}(\tau) = 1$ and $\dot{Z_0}(\tau) > 0$ for every $\tau \in \R$. 
Let $X_\mu$ be the arbitrary location of an observer in Minkowski space. 
We define the corresponding spacetime interval by $R_\mu := X_\mu-Z_\mu$, which gives a function 
$$
     (X,\tau) \mapsto X - Z(\tau) = R(X,\tau).
$$ 

For given $X$ and point particle $Z \col \R \to \R^4$, $\tau \mapsto Z(\tau)$ such that $\sqrt{\sum_{j=1}^3 \dot{Z}_j(\tau)^2} \leq v < 1$ with constant $v$ for all $\tau \in \R$, we define the \emph{retarded proper time} as the unique value $\tau_r(X) \in \R$ such that $Z(\tau_r(X))$ intersects the backward lightcone emanating from $X$ (cf.\ \cite[Section 4.2 and Exercise 4.1]{Parrott}). In particular, we always have $\tau_r(Z(\tau)) = \tau$.  We  obtain the function 
$$
  X \mapsto \widetilde{R}(X) := R(X,\tau_r(X))
$$ 
on Minkowski space, which satisfies
\[
\widetilde{R}_\mu \widetilde{R}^\mu=0,
\]
i.e., $\widetilde{R}(X)$ is a future-directed (nonzero) null four-vector for every $X \in \R^4 \setminus Z(\R)$, i.e., off the world line of the particle, while $\widetilde{R}(Z(\tau)) = R(Z(\tau),\tau) = 0$.

\emph{Notation:}  For a quantity $E$ depending on $(X,\tau)$, as with $R$ above, we will often consider the composition with $X \mapsto (X, \tau_r(X))$ giving a function $\widetilde{E}$ on Minkowski space. 

We define the retarded distance as 
\beq \label{XiZR}
\xi:=\dot{Z}_\mu R^\mu = \dot{Z}^\mu R_\mu.  
\eeq
and  observe that 
\beq\label{xitipo}
   \widetilde{\xi}(X) = \xi(X,\tau_r(X)) = \dot{Z}_\mu(\tau_r(X)) \, R^\mu(X,\tau_r(X)) > 0 \quad \forall X \in \R^4 \setminus Z(\R),
\eeq
since $\dot{Z}$ is timelike and $\tilde{R}$ is lightlike, both future-directed. Note that smoothness of $\widetilde{\xi}$ on $\R^4 \setminus Z(\R)$ follows from that of $X \mapsto \tau_r(X)$, which in turn follows from the implicit function theorem applied to the equation $R_\mu(X,\tau) R^\mu(X,\tau) = 0$ defining $\tau = \tau_r(X)$ under the condition $0 \neq \d_\tau \big( R_\mu(X,\tau) R^\mu(X,\tau) \big) = - 2 R_\mu(X,\tau) \dot{Z}^\mu(\tau) = - 2 \xi(X,\tau)$ and \eqref{xitipo}. For points $X$ on the world line $Z(\R)$ we have $X = Z(\tau)$ with some $\tau \in \R$ and $\widetilde{\xi}(Z(\tau)) = \xi(Z(\tau),\tau) = \dot{Z}_\mu(\tau) \, R^\mu(Z(\tau),\tau) = 0$. Since $R(X,\tau) \to 0 = R(Z(\tau),\tau)$ as $X \to Z(\tau)$ in $\R^4$ we have continuity of $\widetilde{\xi} \col \R^4 \to \R$.

\begin{remark}\label{XiTauRem}  Recall that 
\beq\label{RetRel}
   X_0 - Z_0(\tau_r(X)) = |\vec{X} - \vec{Z}(\tau_r(X))|,
\eeq
where $|.|$ denotes the euclidean norm in $\R^3$ and $\vec{Y}$ the spatial part of $Y$ in Minkowski space. 

(i) Dropping for the moment the argument $\tau_r(X)$ in $Z$ and $\dot{Z}$ for simplicity, assuming $\vec{X} \neq \vec{Z}$, and denoting by $\inp{.}{.}$ the euclidean inner product in $\R^3$, we obtain from \eqref{XiZR} and \eqref{RetRel} the formula
$$
      \widetilde{\xi}(X)  = |\vec{X} - \vec{Z}|  
      \underbrace{\Big( \dot{Z}_0 - \Big< \dot{\vec{Z}} \mid\frac{\vec{X} - \vec{Z}}{ |\vec{X} - \vec{Z}|} \Big> \Big)}_{\qquad =: \, d(X)},
$$
hence $ \widetilde{\xi}(X)$ is proportional to $|\vec{X} - \vec{Z}|$. 
Note that $d(X) > 0$, since the parametrization of the future-directed timelike curve $Z$ by eigentime means $\dot{Z}_0 > 0$ and $\dot{Z}_0^2 - \big|\dot{\vec{Z}}\big|^2 = 1$ and then the Cauchy-Schwarz inequality implies $d(X) \geq \dot{Z}_0 - |\dot{\vec{Z}}| \cdot 1 = \sqrt{1 +  |\dot{\vec{Z}}|^2} -  |\dot{\vec{Z}}| > 0$. 

(ii) We obtain Lipschitz continuity of $\tau_r \col \R^4 \to \R$ from the boundedness assumption $|\dot{\vec{Z}}| \leq v < 1$: Relation \eqref{RetRel} gives 
$Z_0(\tau_r(X)) = X_0 - |\vec{X} - \vec{Z}(\tau_r(X))|$ and recalling the basic inequality between proper time differences and coordinate time differences we may deduce, for any $X, Y \in \R^4$,
\begin{multline*}
   |\tau_r(Y) - \tau_r(X)| \leq |Z_0(\tau_r(Y)) - Z_0(\tau_r(X))| = \Big|Y_0 - |\vec{Y} - \vec{Z}(\tau_r(Y))|
        - X_0 + |\vec{X} - \vec{Z}(\tau_r(X))| \Big|\\
        \leq |Y_0 - X_0| + \Big|  |\vec{X} - \vec{Z}(\tau_r(X))| - |\vec{Y} - \vec{Z}(\tau_r(Y))|  \Big|\\ 
        \leq  |Y_0 - X_0| + |\vec{X} - \vec{Y}| + | \vec{Z}(\tau_r(X)) - \vec{Z}(\tau_r(Y))|
           \leq  |Y_0 - X_0| + |\vec{Y} - \vec{X}| + v \,  |\tau_r(Y) - \tau_r(X)|.
\end{multline*}
This implies $ |\tau_r(Y) - \tau_r(X)| \leq ( |Y_0 - X_0| + |\vec{Y} - \vec{X}|)/(1-v)$, which shows Lipschitz continuity. 

\end{remark}

On $\R^4 \setminus Z(\R)$ we may also consider the null four-vector 
 \[
  \widetilde{K}_\mu := \widetilde{R}_\mu / \widetilde{\xi}
\]
and  the so-called acceleration invariant
\beq\label{KZKappa} 
\widetilde{\kappa}(X) := \ddot{Z}_\mu(\tau_r(X)) \, \widetilde{K}^\mu(X).
\eeq

\subsection{The current of the moving point charge as a distribution on Minkowski space}\label{JDis}
Let the constant $\e$ represent the elementary electric charge, then the current $J = (J_0, J_1, J_2, J_3)$ is a four-vector of distributions $J_\al \in \D'(\R^4)$ (cf.\ \cite[Section 4.1]{Parrott} or \cite[Remark (7.3.25), 2]{Thirring:01}), given by their action on a test function $\phi \in \D(M)$ in the form
$$
    \dis{J_\al}{\phi} := \e \, g_{\al \al} \int_\R \dot{Z}_\al(\tau) \phi(Z(\tau)) \, d \tau.
$$ 
Note that thanks to $\dot{Z}_0(\tau) \geq 1$ we have $|Z_0(\tau) - Z_0(0)| \geq |\tau|$, hence the curve $\tau \mapsto Z(\tau)$ enters and leaves the compact support of $\phi$ in finite time. With $\de_y$ denoting the Dirac distribution at $y \in \R^4$ and the weakly continuous map $y \mapsto \de_y$ from $\R^4$ into  $\D'(\R^4)$ we may interpret $J_\al$ with as a weak distributional integral
$$
    J_\al = \e \, g_{\al \al} \int_\R \dot{Z}_\al(\tau) \de_{Z(\tau)} \, d\tau.
$$
The retarded electromagnetic field generated from this (distributional) current is the (exterior) derivative of the Li\'enard-Wiechert potential
 (cf.\ \cite[Sections 8.2 and 8.3]{Thirring:01} or \cite[Section 4.2]{Parrott}). 

\subsection{The generalized vector field $\bm{\Phi}$} 
The key object in \cite{GEW} is a vector field $\Phi \col \R^4 \to \R^4$ whose d'Alembertian produces terms related to the Li\'enard-Wiechert potential and weak interactions. Following \cite[Equation (4.1)]{GEW}, we introduce the four-vector function $\Phi \col \R^4 \to \R^4$ with components
\beq\label{Phi}
\Phi_\alpha=\frac{\e}{2}\cdot \widetilde{R}_\al \cdot (H \circ \widetilde{\xi}) \quad (\alpha = 0,1, 2, 3),
\eeq
where $H \circ \widetilde{\xi}$ denotes a Colombeau-type generalized function on $\R$, more precisely $H \circ \widetilde{\xi} \in  \G(\R^4)$, and is defined as follows: Let $H \in \G_\infty(\R)$ be associated to the Heaviside function $\theta \in L^\infty(\R)$, i.e., $\theta(s) = 0$ for $s < 0$ and $\theta(s) = 1$ for $s > 0$, thus $H$ is represented by a family of functions $H_\eps \in W^{\infty,\infty}(\R)$, $0 < \eps \leq 1$, with moderate $\eps$-asymptotics and $H \approx \theta$ in the sense that $H_\eps \to \theta$ as $\eps \to 0$ holds in $\D'(\R)$; we necessarily have 
\beq \label{HD}
H' \approx \de_0;
\eeq
in addition, we require $H_\eps(r) = 0$ for $r \leq \eps$, see also Assumption \ref{AssH}(ii), then the composition $H_\eps \circ \widetilde{\xi}$ is smooth on all of $\R^4$ and we let $H \circ \widetilde{\xi}$ denote the class in $\G(\R^4)$ that is represented by the family $(H_\eps \circ \widetilde{\xi})_{0 < \eps \leq 1}$, which satisfies moderateness estimates on compact subsets of $\R^4$: We have globally moderate outer derivatives of $H_\eps$ and inner derivatives of $\widetilde{\xi}$ are relevant only when $\widetilde{\xi}(X) \geq \eps$, hence for $X$ away from the world line $Z(\R)$, due to the support properties of $H_\eps$.

As we will also note in Remark \ref{RemThm}(ii) below, the factor $H \circ \widetilde{\xi}$ in the definition of $\Phi$ is associated (even equal) to the constant function $1$ off the particle's world line, but not equal to $1$ as a generalized function on $\R^4$. This distinction between linear distributional and nonlinear generalized function aspects was a crucial point in Gsponer's modeling of the physics in \cite{GEJP}-\cite{GEW}, where in place of $H$ he employed a generalized function $\Upsilon$. We will justify our choice of $H$ as a  model for Gsponer's $\Upsilon$ in Section \ref{HUps}.

Let $\dal=\d^\mu\d_\mu$ be the d'Alembert operator. We will derive an explicit expression for $\dal\Phi_\alpha$ in Section \ref{SecLWP}. Recall from Equation \eqref{Phi} that we have in detail
\beq\tag{$\text{\ref{Phi}}'$} \label{Phis}
    \Phi_\al(X) = \frac{\e}{2} \cdot R_\al(X,\tau_r(X)) \cdot H(\xi(X,\tau_r(X))),
\eeq
hence $\dal \Phi_\al$ involves implicit differentiation with the retarded proper time $\tau_r$, which is smooth only on $\R^4 \setminus Z(\R)$ and Lipschitz continuous on $\R^4$. Thus, $\Phi_\al$ may be considered a Colombeau-type generalized function on $\R^4 \setminus Z(\R)$ or as a hybrid of a Lipschitz continuous factor times a generalized function from $\G_\infty(\R^4)$. In any case, the representing regularizing family
$$
    X \mapsto \Phi_{\al \eps}(X) := \frac{\e}{2} \cdot R_\al(X,\tau_r(X)) \cdot H_\eps(\xi(X,\tau_r(X)))
$$
consists of Lipschitz continuous functions on $\R^4$ that are smooth on $\R^4 \setminus Z(\R)$. 

\begin{remark} One could alternatively employ regularization via convolution with a family of mollifiers $\rho_\eps$ on $\R^4$ ($0 < \eps \leq 1$) in the form $\Phi_\al^G * \rho_\eps$, where $\Phi_\al^G$ denotes Gsponer's generating function from \cite[Equation (4.1)]{GEW}, and  then study $\dal(\Phi_\al^G * \rho_\eps)$. However, the above variant in \eqref{Phi} does keep the geometrically defined components $\widetilde{R}_\al$ of the spacial interval  fully intact at each value of the regularization parameter
and in the sense of distributional limits both approaches describe the same "generating function".
\end{remark}

\subsection{Implicit differentiation involving the retarded proper time}

Let $E \col \R^4 \times \R \to \R$ be continuously differentiable on $\R^4 \times \R$. As above we consider the related function $\widetilde{E}$ on Minkowski space, i.e., 
$$
  \widetilde{E}(X) = E(X, \tau_r(X)).
$$
Differentiation of $\widetilde{E}$ with respect to $X_\mu$ (or $X^\mu$) will be denoted as usual by $\d_\mu$ (or $\d^\mu$, respectively), while we will use the symbols $\ud_\mu := \d_{X_\mu}$ or $\ud^\mu := \d_{X^\mu}$ for the corresponding derivatives of $E$. 

Recall that $\tau_r$ is smooth on $\R^4 \setminus Z(\R)$ and Lipschitz continuous on $\R^4$ by Remark \ref{XiTauRem}.

\begin{lemma}
On $\R^4 \setminus Z(\R)$ it holds that
\beq\label{3.5}
\d_\mu \widetilde{E} = \widetilde{\ud_\mu E}  + \widetilde{K}^\mu \cdot \widetilde{\d_\tau E},
\eeq
which means in detail
\beq\label{3.5V}
   \d_\mu \Big( E(X,\tau_r(X)) \Big) = \ud_\mu E(X,\tau_r(X)) + K^\mu(X,\tau_r(X)) \cdot \d_\tau E(X,\tau_r(X)).
\eeq
On $\R^4$ we have the analogous formula with $\widetilde{K}^\mu$ replaced by the $L^\infty$ factor $\ud_\mu \tau_r$.
\end{lemma}

\begin{proof} We have from the chain rule
$$
\d_\mu \Big( E(X,\tau_r(X)) \Big)   
  = \ud_\mu E(X,\tau_r(X)) +  \d_\tau E(X,\tau_r(X)) \cdot \ud_\mu \tau_r(X)
$$
and it remains to show that $\ud_\mu \tau_r = \widetilde{K}^\mu$ holds on $\R^4 \setminus Z(\R)$.

Recalling 
\[
X_\nu-Z_\nu(\tau_r(X)) = R_\nu(X,\tau_r(X))
\]
and applying $\ud_\mu$ leads to
\[
\de_{\mu \nu} -\dot{Z}_\nu (\tau_r(X)) \ud_\mu \tau_r (X) = \ud_\mu \Big( R_\nu(X,\tau_r(X)) \Big).
\]
Upon multiplication by $\widetilde{R}^\nu$ and using $\widetilde{R}_\nu \widetilde{R}^\nu=0$ (with summation over $\nu$) we obtain
\[
 \widetilde{R}^\mu - \underbrace{(\dot{Z}_\nu \circ \tau_r) \cdot  \widetilde{R}^\nu}_{\widetilde{\xi}} \cdot \ud_\mu \tau_r 
    = \frac{1}{2}\ud_\mu(\underbrace{\widetilde{R}_\nu \widetilde{R}^\nu}_{=0}),
\]
hence $\widetilde{R}^\mu - \widetilde{\xi} \cdot \ud_\mu \tau_r = 0$ and therefore
\[
 \ud_\mu \tau_r = \frac{\widetilde{R}^\mu}{\widetilde{\xi}} = \widetilde{K}^\mu,
\]
which implies \eqref{3.5}.
\end{proof}

By analogous reasoning we also get
\beq \tag{$\text{\ref{3.5}}'$} \label{3.5s}
\d^\mu \widetilde{E} = \widetilde{\ud^\mu E}  + \widetilde{K}_\mu \cdot \widetilde{\d_\tau E},
\eeq
which we will apply immediately in the proof of (ii) in  the following lemma. We will apply the above formulae in case $E = H_\eps$ for each $\eps$, but often drop the explicit reference to $\eps$ in the notation.

\begin{lemma}\label{XiLem} We have

(i) $\d_\tau\xi  = \xi\kappa-1$ on $\R^4 \setminus Z(\R)$, 

(ii) $\ud^{\mu}\xi=\dot{Z}_\mu$ and $\ud_\mu \xi=\dot{Z}^\mu$, 

(iii) $\nabla \widetilde{\xi}(X) \neq 0$ for all $X \in \R^4 \setminus Z(\R)$.

\end{lemma} 
\proof (i): Recall  $\xi =\dot{Z}_\mu R^\mu =\dot{Z}_\mu (X^\mu-Z^\mu)$ to obtain
$$
  \d_\tau\xi =\ddot{Z}_\mu \underbrace{R^\mu}_{=\xi K^\mu}-\underbrace{\dot{Z}_\mu\dot{Z}^\mu}_{=1}=\xi\kappa-1.
$$

(ii): Follows from $\ud^\mu \xi = \d_{X^\mu}  (\dot{Z}_\al R^\al) = \dot{Z}_\al \, \d_{X^\mu}  (X^\al - Z^\al) = \dot{Z}_\mu$ etc.

(iii): Applying (i) in the second step and (ii) in the third step, we obtain
$$
  \d^\mu \widetilde{\xi} = \widetilde{\ud^\mu \xi} + \widetilde{K}_\mu  \widetilde{\d_\tau \xi} 
  = \widetilde{\ud^\mu \xi} + \widetilde{K}_\mu (\widetilde{\xi} \widetilde{\kappa} - 1)
  =  \widetilde{\dot{Z}^\mu} + \frac{\widetilde{R}_\mu}{\widetilde{\xi}} (\widetilde{\xi} \widetilde{\kappa} - 1)
  =  \widetilde{\dot{Z}^\mu} + \widetilde{R}_\mu \widetilde{\kappa} - \frac{\widetilde{R}_\mu}{\widetilde{\xi}}.
$$
Therefore, the equation $ \d^\mu \widetilde{\xi} =0 $ is equivalent to 
$$
     \widetilde{\dot{Z}^\mu} = \left( \frac{1}{\widetilde{\xi}} -  \widetilde{\kappa} \right) \widetilde{R}_\mu,
$$
which is impossible to obtain, because $\dot{Z}$ is timelike while $\widetilde{R}$ is null. 
\qed

\section{Emergence of the Li\'{e}nard-Wiechert potential from the d'Alembertian of $\Phi$}\label{SecLWP}

To reduce notational overload in the detailed calculations below we will from now on write $R$ in place of $\widetilde{R}$ and similarly for the other quantities. Moreover,  function arguments other than $\xi$ will often be omitted and differentiation will be understood to be carried out on $\R^4 \setminus Z(\R)$ whenever it involves $\xi$ directly.

\textbf{Some useful identities:} 
The following properties will be exploited in the calculations below:

\begin{enumerate}
	\item\label{RZ}  Properties of $R$: Recall that \label{RXiK} $\Ra=\xi K_\alpha$ by definition. For the first derivatives, we clearly  have $\d_\tau R_\alpha=-\dot{Z}_\alpha$; furthermore, $\ud_\mu \Ra = \delta_{\mu\alpha}$ and $\ud^\mu R_\mu = -2$, 
	since $\d_{X^\mu} (X_\mu - Z_\mu) = \d_{X^\mu} (g_\mu^\mu X^\mu - Z_\mu) = \sum_{\mu = 0}^3 g_\mu^\mu = -2$; 
	similarly, $\ud^\mu R^\mu = 4$, by $\d_{X^\mu} (X^\mu - Z^\mu) = \sum_{\mu = 0}^3 1 = 4$.  \label{RDelta} 
	For the second derivatives we obtain \label{BoxR} $\ud^\mu \ud_\mu \Ra := \underline{\dal} \Ra =0$.

	\item\label{RK} The relation $\ud^\mu\Ra K_\mu=K_\alpha$ follows from \ref{RDelta}).
		
	\item \label{ZZ} Parametrization of $Z$ by eigentime means $\dot{Z}_\mu\dot{Z}^\mu=1$ and 		
	\label{KZ} $K_\mu \dot{Z}^\mu=1$ follows by definition. Moreover, \label{KK} 
	from $R_\mu R^\mu = 0$ it also holds that $K_\mu K^\mu = 0$.
		
	\item\label{dZ} Since $\dot{Z}_\alpha$ does not depend on X, we obtain $\ud_\mu\dot{Z}_\alpha=0$.

	\item Derivatives of $\xi$: Direct calculation gives \label{RXiZ} $\ud_\mu \Ra \, \ud^\mu\xi = \dot{Z}_\alpha$, 	
	\label{ddXi} $\ud_\mu \d_\tau\xi = \d_\tau \ud_\mu \xi = \ddot{Z}^\mu$ follows directly from \eqref{XiZR}, and, obviously,
	 \label{BoxXi} $\underline{\dal} \xi=0$.
	
	\item\label{dmuKmu} We have $\d^\mu K^\mu = \frac{2}{\xi}$: This follows from \ref{RZ}), (i) and (ii) in 
	   Lemma \ref{XiLem},  Equation \eqref{3.5s},  and the definition of $K^\mu$; indeed
\begin{multline*}
  \d^\mu K^\mu = \d^\mu \left( \frac{R^\mu}{\xi} \right) = \frac{\d^\mu R^\mu \cdot \xi - R^\mu \cdot \d^\mu \xi}{\xi^2}
  = \frac{1}{\xi^2} \big( (\ud^\mu R^\mu + K_\mu \d_\tau R^\mu) \xi - R^\mu (\ud^\mu \xi + K_\mu \d_\tau \xi) \big) \\
  = \frac{1}{\xi^2} \big( (4 - K_\mu \dot{Z}^\mu) \xi - R^\mu (\dot{Z}_\mu + K_\mu (\xi \kappa - 1))  \big) 
  =  \frac{1}{\xi^2} \big( 4 \xi - \underbrace{K_\mu \dot{Z}^\mu}_{\frac{R_\mu \dot{Z}^\mu}{\xi} = 1} \xi - 
      \underbrace{R^\mu  \dot{Z}_\mu}_{\xi} + \underbrace{R^\mu K_\mu}_{0} (\xi \kappa - 1)  \big) \\
  =  \frac{1}{\xi^2} ( 4 \xi - \xi - \xi) = \frac{2}{\xi}.
\end{multline*}
	
\end{enumerate}

\textbf{First-order derivatives of $\bm{\Phi_\al}$:} Recall that the generalized function factor occurring in the definition of $\Phi_\al$ in Equations \eqref{Phi} or \eqref{Phis} is represented by the family of smooth functions $X \mapsto H_\eps(\xi(X,\tau_r(X)))$ with $0 < \eps \leq 1$. Using short-hand notations and writing here temporarily $H$ or $H(\xi)$ to mean $H_\eps \circ \widetilde{\xi}$, we have
\begin{align*}
\partial_\mu\Phi_{\alpha}
&=\frac{\e}{2}\partial_\mu R_\alpha H+\frac{\e}{2}R_\alpha\partial_\mu H\\
&=\frac{\e}{2}\left(\ud_\mu R_\alpha+K^\mu\partial_{\tau}R_\alpha \right) H
  +\frac{\e}{2}R_\alpha H'\cdot \partial_\mu\xi\\
&=\frac{\e}{2}\left(\ud_\mu R_\alpha+K^\mu\partial_{\tau}R_\alpha \right) H
    +\frac{\e}{2}R_\alpha H' \cdot \left(\ud_\mu \xi+K^\mu\partial_\tau\xi\right),
\end{align*}
by making use of Equation \eqref{3.5}.

\textbf{Second-order derivatives of $\bm{\Phi_\al}$:}
In a first step we calculate
\begin{align*}
\frac{2}{e}\dal\Phi_{\alpha}
& = \frac{2}{e}\d^\mu ( \d_\mu \Phi_{\alpha}) = 
   \d^\mu \Big(  \left(\ud_{\mu} R_\alpha+K^\mu\partial_{\tau}R_\alpha \right) H
     +R_\alpha H' \cdot \left(\ud_{\mu}\xi+K^\mu\partial_\tau\xi\right) \Big) \\
&=\underbrace{\d^\mu\left(\ud_{\mu}\Ra\cdot H\right)}_{A}+\underbrace{\d^\mu\left(K^\mu\d_\tau\Ra\cdot H\right)}_{B}
 +\underbrace{\d^\mu\left(\Ra H' \cdot \ud_{\mu}\xi\right)}_{C}+\underbrace{\d^\mu\left(\Ra H' K^\mu\d_\tau\xi\right)}_{D}.
\end{align*}

We investigate the four summands separately, keeping in mind that $H$ is short-hand for a generalized function on Minkowski space involving a family regularizations and the retarded proper time. We exercise some caution and provide all details in the otherwise elementary calculations to follow because it is a mix of smooth and generalized functions with direct or more implicit differentiations whenever the chain rule involves the retarded proper time:

The first term expands to 
\begin{align*}
A
&=\d^\mu(\ud_{\mu}\Ra\cdot H)=\d^\mu(\ud_{\mu}\Ra)\cdot H+\ud_{\mu}\Ra\d^\mu H\\
&=\ud^\mu \ud_{\mu}\Ra\cdot H+K_\mu\d_\tau\ud_{\mu}\Ra\cdot H+\ud_{\mu}\Ra H'\d^\mu \xi\\
&=\underline{\dal} \Ra\cdot H+K_\mu\ud_{\mu}\d_\tau\Ra\cdot H+\ud_{\mu}\Ra \ud^\mu \xi \cdot H' (\xi)+\ud_{\mu}\Ra K_\mu\d_\tau\xi \cdot H' (\xi),
\end{align*}
where we again used Equation \eqref{3.5}. Recall that $H'$ will represent certain pullbacks of Dirac delta $\de_0$ to Minkowski space in the distributional limit. The expression for term labeled $B$ in all its details is 
\begin{align*}
B
&=\d^\mu\left(K^\mu\d_\tau\Ra\cdot H\right)=\d^\mu K^\mu\cdot\d_\tau\Ra\cdot H+K^\mu\cdot\d^\mu (\d_\tau\Ra)\cdot H +K^\mu\cdot\d_\tau\Ra\cdot H'\d^\mu \xi\\
&= \frac{2}{\xi} \cdot\d_\tau\Ra\cdot H + 
     K^\mu \cdot (\ud^\mu \d_\tau\Ra+K_\mu\d_\tau^2 \Ra)\cdot H \\ 
    & \hphantom{=}  +  K^\mu\cdot\d_\tau\Ra\cdot H'(\xi)(\ud^\mu \xi+K_\mu\d_\tau\xi)\\
& = \Big( \frac{2}{\xi} \cdot\d_\tau\Ra+K^\mu\ud^\mu\d_\tau\Ra \Big) \cdot H+K^\mu\cdot\d_\tau\Ra\cdot\ud^\mu \xi \cdot H'(\xi),
\end{align*}
by using Equation \eqref{3.5} and identities from \ref{KK}) and \ref{dmuKmu}). The third term $C$ will be the first to involve $H''$, hence pullbacks of $\de_0'$ in the distributional limit, and  reads
\begin{align*}
C
&=\d^\mu\left(\Ra H'(\xi)\cdot\ud_{\mu}\xi\right)=(\d^\mu\Ra)\cdot H'(\xi)\cdot\ud_{\mu}\xi+\Ra\cdot\d^\mu H'(\xi)\cdot\ud_{\mu}\xi+\Ra\cdot H'(\xi)\cdot\d^\mu(\ud_{\mu}\xi)\\
&=(\ud^\mu\Ra+K_\mu\d_\tau\Ra)\cdot H'(\xi)\cdot\ud_{\mu}\xi +\Ra\cdot H''(\xi) \cdot \d^\mu(\xi) \cdot\ud_{\mu}\xi+\Ra\cdot H'(\xi)\cdot(\ud^\mu \ud_{\mu}\xi+K_\mu\d_\tau\ud_{\mu}\xi)\\
&=(\ud^\mu \Ra\ud_{\mu}\xi+K_\mu\d_\tau\Ra\ud_{\mu}\xi+\Ra(\underline{\dal} \xi+K_\mu\ud_{\mu}\d_\tau\xi))H'(\xi)
 + \Ra(\ud^{\mu} \xi \ud_{\mu} \xi + K_\mu\ud_{\mu}\xi\d_\tau\xi)H''(\xi) .
\end{align*}
Finally, the last summand defining $D$ gives, by calling on \ref{ZZ}),
\begin{align*}
D
&=\d^\mu\left(\Ra H'(\xi)K^\mu\d_\tau\xi\right)= (\d^\mu\Ra)H'(\xi)K^\mu\d_\tau\xi+ \Ra(\d^\mu H'(\xi))K^\mu\d_\tau\xi\\
&+\Ra H'(\xi)(\d^\mu K^\mu)\d_\tau\xi+\Ra H'(\xi)K^\mu(\d^\mu\d_\tau\xi)\\
&=(\ud^\mu\Ra+K_\mu\d_\tau\Ra)H'(\xi)K^\mu\d_\tau\xi+\Ra H''(\xi)(\ud^\mu\xi+K_\mu\d_\tau\xi)K^\mu\d_\tau\xi\\
&+ \Ra H'(\xi)(\ud^\mu K^\mu+K_\mu\d_\tau K^\mu)\d_\tau\xi+\Ra H'(\xi)K^\mu(\ud^\mu \d_\tau\xi+K_\mu\d^2_\tau\xi)\\
&= ((\ud^\mu\Ra K^\mu+\ud^\mu K^\mu \Ra+K_\mu\d_\tau K^\mu\Ra)\d_\tau\xi +\Ra K^\mu\ud^\mu \d_\tau\xi)H'(\xi)\\
&+\Ra\ud^\mu \xi K^\mu\d_\tau\xi H''(\xi).
\end{align*} 

We collect everything and define the terms $E$, $F$, and $G$ upon sorting by derivatives of $H$, i.e.,
\[
A+B+C+D=E\cdot H(\xi)+F\cdot H'(\xi)+G\cdot H''(\xi).
\]

Using \ref{BoxR}) the term $E$  simplifies considerably, since
$$
E =\underbrace{\underline{\dal} \Ra}_{0} +
   K_\mu \cdot \underbrace{\ud_{\mu}\d_\tau\Ra}_{0} +
     \frac{2}{\xi} \cdot \d_\tau\Ra+
     K^\mu \cdot \underbrace{\ud^\mu\d_\tau\Ra}_{0}
  = - \frac{2}{\xi} \dot{Z}_\alpha
$$

The lengthy expression for $F$ can be simplified by Lemma \ref{XiLem}(ii), \ref{RK}), \ref{dZ}), and \ref{RXiZ}), and yields
\begin{align*}
F
&=\ud_{\mu}\Ra \ud^\mu\xi+\ud_{\mu}\Ra K_\mu\d_\tau\xi+K^\mu\ud^\mu \xi\d_\tau\Ra+\ud^\mu\Ra \ud_{\mu}\xi+K_\mu \ud_{\mu}\xi \d_\tau\Ra + \Ra \underline{\dal} \xi+\Ra K_\mu \ud_{\mu}\d_\tau\xi\\
&+\ud^\mu\Ra K^\mu\d_\tau\xi+\ud^\mu K^\mu\Ra\d_\tau\xi+K_\mu\d_\tau K^\mu\Ra\d_\tau\xi+\Ra K^\mu\ud^\mu\d_\tau\xi\\
&=\Za+K_\alpha\d_\tau\xi-\Za+\Za-\Za + \Ra \underline{\dal} \xi+\Ra K_\mu\ud_{\mu}\d_\tau\xi\\
&+K_\alpha\d_\tau\xi+\ud^\mu K^\mu\Ra\d_\tau\xi+K_\mu\d_\tau K^\mu\Ra\d_\tau\xi+\Ra K^\mu\ud^\mu\d_\tau\xi\\
&=2\cdot K_\alpha\d_\tau\xi+\Ra K_\mu\ud_{\mu}\d_\tau\xi+\ud^\mu K^\mu\Ra\d_\tau\xi
 +K_\mu\d_\tau K^\mu\Ra\d_\tau\xi + \Ra K^\mu\ud^\mu\d_\tau \xi.
\end{align*}
Upon further inspection of this last expression and calling on \ref{ddXi}), \ref{dmuKmu}), and Lemma \ref{XiLem}(i), we can obtain a much shorter description by calculating
\begin{align*}
F
&=
2K_\alpha\d_\tau\xi+\Ra K_\mu\ud_{\mu}\d_\tau\xi+\ud^\mu K^\mu\Ra\d_\tau\xi+K_\mu\d_\tau K^\mu\Ra\d_\tau\xi+\Ra K^\mu\ud^\mu\d_\tau\xi\\
&= \Big( 2K_\alpha+ \underbrace{(\ud^\mu K^\mu + K_\mu\d_\tau K^\mu )}_{2/\xi} \Ra \Big) \d_\tau\xi 
+ \Ra K_\mu\ud_{\mu}\d_\tau\xi + \Ra \underbrace{K^\mu\ud^\mu\d_\tau\xi}_{K_\mu\ud_{\mu}\d_\tau\xi}\\
&=2 K_\alpha\d_\tau\xi   + \frac{2 \Ra\d_\tau\xi}{\xi}  + 
  2 \Ra K_\mu \cdot  \underbrace{ \ud_{\mu} \d_\tau \xi}_{\ddot{Z}^\mu} \\
&=  \left(2 K_\alpha + \frac{2 \Ra}{\xi} \right) \d_\tau\xi  + 2 \Ra  \underbrace{K_\mu \ddot{Z}^\mu}_{\kappa} =
   \left(2 K_\al + 2 K_\al \right) (\xi \kappa - 1) + \, 2 \underbrace{\Ra}_{K_\al \xi} \kappa\\
   &= 6 K_\al \xi \kappa - 4 K_\al.
\end{align*}

Finally, it remains to find an improved representation of $G$. Using (i) and (ii) from Lemma \ref{XiLem}, we get
\begin{align*}
G
&=\Ra \underbrace{\ud^{\mu} \xi \ud_{\mu} \xi}_{\dot{Z}_\mu \dot{Z}^\mu = 1} 
   + \Ra K_\mu\ud_{\mu}\xi\d_\tau\xi+\Ra\ud^\mu\xi K^\mu\d_\tau\xi
 = \Ra + \Ra \underbrace{K_\mu \dot{Z}^\mu}_{\xi/\xi = 1} \d_\tau\xi + \Ra \underbrace{\dot{Z}_\mu K^\mu}_{\xi/\xi = 1} \d_\tau\xi\\
&= \Ra + 2 \Ra \cdot \d_\tau\xi =  K_\alpha\xi (2 \xi\kappa-1).
\end{align*}

Summarizing, we arrive at the following result. 
\begin{theorem} The d'Alembertian of $\Phi$ is given on $\R^4 \setminus Z(\R)$ for each component $\al=0,1,2,3$ by
%\begin{multline} 
\beq \label{BoPhi}
   \dal\Phi_\alpha=\frac{e}{2}\Big( - \frac{2 \Za}{\xi} \cdot H(\xi) + 
   K_\alpha(6 \xi\kappa - 4) H'(\xi) +  \xi K_\alpha(2 \xi\kappa-1)H''(\xi)\Big) =  - \Lambda_\al  H(\xi) +  \Psi_\al,
\eeq
where we have defined
$$
     \Lambda_\al :=  e  \frac{\Za}{\xi} \quad\text{and}\quad  \Psi_\al := e K_\alpha 
        \left( (3 \xi\kappa - 2) H'(\xi) +  (\xi\kappa- \frac{1}{2}) \xi H''(\xi)\right).
$$
%\end{multline}
\end{theorem}

Recall that we temporarily simplified the notation, writing $\xi$ in place of $\widetilde{\xi}$ etc., which also applies to the above equation. We will now switch back to the original notation with $\widetilde{\;\;\;}$.

We observe that the right-hand side of \eqref{BoPhi} allows for a natural extension to an object in $\G(\R^4)$: For every $\eps \in\, ]0,1]$, $H_\eps \circ \widetilde{\xi}$ vanishes for all $X$ with $\widetilde{\xi}(X) \leq \eps$, hence both terms $- \Lambda_\al \,  (H_\eps \circ \widetilde{\xi})$ and $\Psi_{\al,\eps}$ can be extended smoothly by $0$ (recall that also $\widetilde{\kappa}$ involves a term $1/\widetilde{\xi}$) and provide a moderate net $(- \Lambda_\al \,  (H_\eps \circ \widetilde{\xi}) + \Psi_{\al,\eps})_{0 < \eps \leq 1}$ when estimated on compact subsets of $\R^4$. We may thus consider
\beq\label{LamPsiGlobal}
    \Lambda_\al  H(\xi) \in \G(\R^4) \quad\text{and}\quad \Psi_\al \in \G(\R^4).
\eeq

\begin{remark}\label{RemThm} (i) As was already noted in \cite[Section 4]{GEW}, the factor $\Lambda_\al$ in the first term on the right-hand side of Equation  \eqref{BoPhi} is the  Li\'enard-Wiechert potential, which is usually derived from the distributional current given in Subsection \ref{JDis} (cf.\ \cite[Sections 8.2 and 8.3]{Thirring:01} or \cite[Section 4.2]{Parrott}). 
(Beware of the fact that some sources use a convention with the reversed signature of the Minkowski metric.) 
As mentioned in \cite[paragraph after Equation (4.2)]{GEW}, a related classical result in the form $\dal R = 2/\xi$ has apparently been observed in the context of using retarded coordinates.

(ii) Recall that $H_\eps \to \theta$ as $\eps \to 0$ in the sense of distributions, as mentioned below Equation \eqref{Phi},
and that the condition $H_\eps(r) = 0$ for $r \leq \eps$ guarantees smoothness of the compositions $H_\eps \circ \widetilde{\xi}$, $H_\eps' \circ \widetilde{\xi}$, and $H_\eps'' \circ \widetilde{\xi}$, which appear on the right-hand side of Equation \eqref{BoPhi}.
Since $\nabla \widetilde{\xi}$ vanishes nowhere on $\R^4 \setminus Z(\R)$ by Lemma \ref{XiLem}(iii), we may use the pullback of distributions on $\R$ by $\widetilde{\xi}$ to distributions on $\R^4 \setminus Z(\R)$ according to \cite[Theorem 6.1.2]{Hoermander:V1}. This gives a continuous map $\D'(\R) \to \D'(\R^4 \setminus Z(\R))$, $u \mapsto \xip u$, where in case of $u$ being continuous, $\xip u = u \circ \widetilde{\xi}$. 
In particular, $H_\eps \circ \widetilde{\xi} = \xip H_\eps \to \xip \theta$ as $\eps \to 0$. The detailed representation of the pullback in \cite[Equation (6.1.1)]{Hoermander:V1} shows $\xip \theta = 1$, hence we conclude
$$
   H_\eps (\widetilde{\xi}) \to 1 \quad\text{and}\quad - \Lambda_{\al}  H_\eps (\widetilde{\xi}) \to - \Lambda_\al \
   \quad\text{as $\eps \to 0$ in } \D'(\R^4 \setminus Z(\R)).
$$
\end{remark}

Gsponer in \cite[page 7]{GEW} mentions that in a pure distribution theoretic interpretation of \eqref{BoPhi} only the term with the Li\'enard-Wiechert potential is visible. We provide here more support for his observation by analyzing the generalized functions established in  \eqref{LamPsiGlobal}, in which we make the additional technical assumption that $H_\eps(r) = 1$ holds for $r \geq 2 \eps$, compare also with Assumption \ref{AssH}(iii) below.

\begin{theorem} Off the world line of the charged point particle, i.e., restricted to $\R^4 \setminus Z(\R)$, we have $\Psi_\al = 0$ and $\dal \Phi_\al = - \Lambda_\al$. 
\end{theorem}

\proof On any compact subset of $\R^4 \setminus Z(\R)$, the function $\widetilde{\xi}$ has a positive lower bound yielding $H_\eps \circ \widetilde{\xi} = 1$, $H_\eps' \circ \widetilde{\xi} = 0$, and $H_\eps'' \circ \widetilde{\xi} = 0$ there for all sufficiently small $\eps > 0$. Therefore $- \Lambda_\al \,  (H_\eps \circ \widetilde{\xi}) = - \Lambda_\al$ and $\Psi_{\al,\eps} = 0$, hence Equation \eqref{BoPhi} implies   $\dal \Phi_{\al,\eps} = - \Lambda_\al \,  (H_\eps \circ \widetilde{\xi}) +  \Psi_{\al,\eps} = - \Lambda_\al$.
\qed

\begin{remark} 
(i) Although it is physically reasonable to have potentials defined off the world line of the charged particle only (see, e.g., \cite[(8.3.1)]{Thirring:01}), we can here make use of the framework of generalized functions in $\G(\R^4)$ and observe we might have $\Psi_\al \neq 0$, while $\Psi_\al \approx 0$. This can indeed be shown with a regularization $(H_\eps)_{0 < \eps \leq 1}$ as specified in \eqref{Remedy} below for the special case of the static world line $Z(\tau) = (\tau, 0,0,0)$: Denoting $r := |\vec{X}|$ we obtain $\widetilde{\xi}(X) = r$,  $\widetilde{K_0}(X) = 1$, and $\widetilde{\kappa}(X) = 0$. This yields 
$$
   \Psi_{0,\eps}(X) = -e \Big( 2 H_\eps'(r) + \frac{r}{2} H_\eps''(r) \Big)
$$
and the structure of the regularization $H_\eps$ gives $\supp(\Psi_{0,\eps}) \subseteq [\eps, 2 \eps]$ and $\| \Psi_{0,\eps} \|_\infty = O(1/\eps)$.

To show  $\Psi_\al \approx 0$, we have to integrate against a test function on $\R^4$, where the spatial integration is restricted to a volume of size $O(\eps^3)$, hence can be estimated by $\| \Psi_{0,\eps} \|_\infty \cdot O(\eps^3) = O(\eps^2)$.

We claim that $\Psi_\al \neq 0$, i.e., $(\Psi_{0,\eps})_{0 < \eps \leq 1}$ is nonnegligible, if $e \neq 0$. We have
$$
   \Psi_{0,\eps}(X) = -e \Big( \frac{2}{\eps} F'(r/\eps) + \frac{r}{2 \eps^2} F''(r/\eps) \Big) =
   - \frac{e}{\eps} \left. \Big( 2 F'(s) + \frac{s}{2} F''(s) \Big) \right|_{s = r/\eps}.
$$
Arguing by contradiction, negligibility of $ \Psi_{0,\eps}$ would give bounds $|\Psi_{0,\eps}(X)| = O(\eps^q)$ for every $q \in \N$ and $X$ in some compact set. This imposes $0 = 2 s^3 (2 F'(s) + \frac{s}{2} F''(s)) =(s^4 F'(s))'$ identically on $]1,2[$, contradicting the specifications of $F$.

 (ii) A main discussion in \cite[Sections 4 and 5]{GEW}  (where $\Psi_\al$ is denoted as $\Phi^K$) stresses the possibility of $\Psi_\al \neq 0$ as a generalized function, although $\Psi_\al \approx 0$ in the distributional sense. The emphasis is then on physical arguments relating a renormalized variant of $\Psi_\al$ to weak interaction in the simplified Fermi theory. 
 
\end{remark}

\section{Regularization of the electromagnetic self-interaction energy of a point particle at rest}\label{SecSelf}

We consider now the situation of a charged particle in its rest frame. The static case of the Li\'{e}nard-Wiechert potential yields the Coulomb potential (cf.\ \cite[(8.2.17) and (8.2.18)]{Thirring:01}). Note that, in principle, electrodynamics in the form of Maxwell's equations can be derived from electrostatics plus the requirement of Lorentz invariance (cf.\ \cite{Scharf}). 
As in \cite{GsponerJMP}, we discuss a regularization approach in describing the electric monopole, magnetic dipole, electron singularity, and self-energy. Our analysis is more rigorous and for a class of regularizations not restricted to convolution techniques.

To begin with, we rewrite the potential \cite[Equation (6.1)]{GsponerJMP} as generalized function on $\R^3$ in the form
\beq\label{el_pot}
 \phi(x) = \e \frac{H(|x|)}{|x|},
\eeq
where, similarly as with \eqref{Phi}, $H$ denotes a generalized function in $\G_\infty(\R)$, replacing Gsponer's $\Upsilon$, associated to the Heaviside function $\theta$ as shown in Section \ref{HUps}, so that $\phi$ is associated, but not equal, to the Coulomb potential. 

\begin{assumption}\label{AssH}
In more detail, we assume $H$ to be represented  by a $\G_\infty$-moderate family $(H_\eps)_{\eps \in (0,1]}$ of smooth functions $H_\eps \col \R \to \R$ with the following properties (slightly different from the definition in \cite[Subsection 3.3.1]{Colombeau:92}): 
\begin{enumerate}

\item[(i)] $H_\eps \geq 0$ and $H_\eps' \geq 0$,

\item[(ii)] $H_\eps(r) = 0$ if $r \leq \eps$, 

\item[(iii)] $H_\eps(r) = 1$ if $r \geq 2 \eps$.

\end{enumerate}
\end{assumption}

Note that (ii) enables us to define the right-hand side in \eqref{el_pot} via the smooth representatives $x \mapsto H_\eps(|x|)/|x|$, $\R^3 \to \R$. An example of such a regularizing family can be achieved in the form 
$$
   H_\eps(r) := \int_\R \theta(r-s) \frac{1}{\eps} \chi\left(\frac{s}{\eps}\right) \, d s =  
       \int_0^r \frac{1}{\eps} \chi\left(\frac{s}{\eps}\right) \, d s = \int_0^{r/\eps} \chi(t) \, d t
$$ 
with a mollifier $\chi \in \S(\R)$ such that $\int_\R \chi(s) \, ds = 1$, $\chi \geq 0$, and $\supp \chi \subseteq [1,2]$.

We may further calculate in the usual way, either with representatives or Colombeau classes, and obtain the electric field (\cite[Equation (6.2)]{GsponerJMP})
$$
   E(x) = - \nabla \phi(x) = \e \left( \frac{H(|x|)}{|x|^2}  - \frac{H'(|x|)}{|x|}\right) \frac{x}{|x|}
$$ 
as well as the charge density $\rho$ (\cite[Equation (6.3)]{GsponerJMP}) following from
$$
  4 \pi \rho(x) = \div E(x) = - \e \frac{H''(|x|)}{|x|}.
$$

\begin{proposition} We have $\rho \approx \e \de_0$.
\end{proposition}

\proof The action of $\rho_\eps$ on a test function $\vphi \in \D(\R^3)$ is given by
\begin{multline*}
  - \frac{4 \pi}{\e} \dis{\rho_\eps}{\vphi} =  \int_{\R^3} \frac{H_\eps''(|x|)}{|x|} \vphi(x) \, dx 
  = \int_0^\infty \int_{S^2} \frac{H_\eps''(r)}{r} \vphi(r \om) \, r^2 \, d\om \, d r
  = \int_0^\infty r H_\eps''(r) \overbrace{\int_{S^2} \vphi(r \om) \, d \om}^{(M \vphi)(r)} \, d r\\
  =  \underbrace{ \left.  H_\eps'(r) r  (M \vphi)(r)  \right|_0^\infty}_{0} 
  - \;    \int_0^\infty H_\eps'(r) \big( r (M\vphi)(r)\big)' \, dr\\ 
  =    - \int_0^\infty H_\eps'(r) \big( (M\vphi)(r) + r (M\vphi)'(r) \big) \, dr.
\end{multline*}
Note that $M\vphi$ can be extended to $r < 0$ as an even function (by symmetry of $S^2$) and that $\supp H_\eps \subseteq [\eps,\infty[$ (by property (ii) above), hence we obtain
$$
   - \frac{4 \pi}{e} \dis{\rho_\eps}{\vphi} =  -  \int_{-\infty}^\infty H_\eps'(r) \big( (M\vphi)(r) + r (M\vphi)'(r) \big) \, dr
   = -  \dis{H_\eps'}{M\vphi + r (M\vphi)'}.
$$
Since $H_\eps' \to \de_0$ as $\eps \to 0$, we arrive at
$$
    - \frac{4 \pi}{\e} \dis{\rho_\eps}{\vphi} \to -  (M\vphi)(0) = - \int_{S^2} \vphi(0) \, d\om = 
    - 4 \pi \vphi(0) = - 4 \pi \dis{\de_0}{\vphi}.
$$
\qed

We want to define the electric self-energy of the point charge by integration of the generalized energy density $|E|^2 \in \G(\R^3)$, whose expression in terms of $H$ is 
$$
   |E(x)|^2 = \frac{\e^2}{|x|^2} \left( \frac{H(|x|)}{|x|} - H'(|x|) \right)^2
   = \underbrace{\frac{\e^2 H(|x|)^2}{|x|^4}}_{\E_1(x)} + 
   \underbrace{H'(|x|) \frac{\e^2}{|x|^2} \Big( H'(|x|) -  \frac{2 H(|x|)}{|x|} \Big)}_{\E_2(x)}, 
$$
where $\E_2 \in \G(\R^3)$ has compact support and, upon choosing a compact ball $K_R$ of radius $R > 0$ (containing the interior of this support), we may define (cf.\ \cite[Definition 1.2.60]{GKOS:01}) the generalized number
$$
     \int_{\R^3} \E_2(x) \, dx := \int_{K_R} \E_2(x) \, dx.
$$
As for $\E_1 \in \G(\R^3)$, we note that $\E_1(x) = e^2 / |x|^4$ when $|x|  > 1$. Choose a nonnegative cut-off function $\chi \in \D(\R^3)$ with $\chi = 1$ near the unit ball and write 
$$
  \E_1(x) = \chi (x) \E_1(x) + \frac{e^2 (1 - \chi(x))}{|x|^4},
$$
where the first term on the right-hand side is a compactly supported generalized function and the second term is an integrable smooth function on $\R^3$. We put 
$$
   \int_{\R^3} \E_1(x) \, dx := \int_{\R^3} \chi(x) \E_1(x) \, dx + \int_{\R^3} \frac{e^2 (1 - \chi(x))}{|x|^4} \, dx
$$
and claim that this defines a generalized number independent of the chosen cut-off $\chi$. Indeed, if $\vphi \in \D(\R^3)$ is also nonnegative and satisfies $\vphi = 1$ in a neighborhood of the unit ball, then we have 
\begin{multline*}
   \int_{\R^3} \chi(x) \E_1(x) \, dx + \int_{\R^3} \frac{e^2 (1 - \chi(x))}{|x|^4} \, dx - 
     \int_{\R^3} \vphi(x) \E_1(x) \, dx - \int_{\R^3} \frac{e^2 (1 - \vphi(x))}{|x|^4} \, dx\\
     = \int_{\R^3} \big(\chi(x) - \vphi(x)\big) \E_1(x) \, dx - \int_{\R^3} \frac{e^2 (\chi(x) - \vphi(x))}{|x|^4} \, dx
\end{multline*}
and it remains to observe that, as a generalized number, $\int_{\R^3} \big(\chi(x) - \vphi(x)\big) \E_1(x) \, dx = \int_{\R^3} \frac{e^2 (\chi(x) - \vphi(x))}{|x|^4} \, dx$, since $ \E_1(x) = e^2/|x|^4$ on $\supp(\chi - \vphi)$.

We may now define the \emph{electric self-energy} as the generalized number 
$$
   \Uel := \frac{1}{8 \pi} \int_{\R^3} |E(x)|^2 \, dx :=  \frac{1}{8 \pi}  \Big( \int_{\R^3} \E_1(x) \, dx +  \int_{\R^3} \E_2(x) \, dx \Big).
$$

The support properties of $H_\eps(|x|)$ and $H_\eps'(|x|)$ and spherical coordinates allow us to write
$$
    \int_{\R^3} |E_\eps(x)|^2 \, dx 
   = 4 \pi \e^2 \int_0^\infty \frac{H_\eps(r)^2}{r^2} \, dr  + 
   4 \pi \e^2 \int_0^\infty H_\eps'(r)^2 \, dr - 
   8 \pi \e^2 \int_0^\infty \frac{H_\eps(r) H_\eps'(r)}{r}  \, dr.
$$
Integration by parts in the last term gives$ \int_0^\infty H_\eps(r) H_\eps'(r)/r  \, dr =  \int_0^\infty H_\eps(r)^2/r^2 \, dr/2$, hence a  representative of $\Uel$is given by
$$
     \Uel^\eps =  \frac{1}{8 \pi} \int_{\R^3} |E_\eps(x)|^2 \, dx =  \frac{e^2}{2} \int_0^\infty H_\eps'(r)^2 \, dr.
$$

\begin{theorem}\label{UelThm} If $e \neq 0$, then the generalized number $\Uel$ is infinite in the sense that
$$
   \Uel^\eps \to \infty \quad (\eps \to 0).
$$
\end{theorem}
\proof By assumptions (ii) and (iii) on $H_\eps$ we have $H_\eps' (r) = 0$ if $r < \eps$ or $r > 2 \eps$, hence
$$
 a_\eps : = \frac{2}{\e^2} \Uel^\eps = \int_\eps^{2 \eps} H_\eps'(r)^2 \, dr \geq 0.
$$ 
Recalling  $H_\eps(\eps) = 0$ and $H_\eps(2 \eps) = 1$, we have
$$
   1 = H_\eps(2 \eps)  - H_\eps(\eps) = \int_\eps^{2 \eps} H_\eps'(r) \, dr. 
$$
Thus we obtain from the Cauchy-Schwarz inequality,
$$
   a_\eps = \int_\eps^{2 \eps} H_\eps'(r)^2 \, dr = \frac{1}{\eps} \int_\eps^{2 \eps} 1^2 \, dr \cdot \int_\eps^{2 \eps} H_\eps'(r)^2 \, dr 
   \geq \frac{1}{\eps} \cdot \Big( \int_\eps^{2 \eps} (1 \cdot H_\eps'(r)) \, dr \Big)^2 = \frac{1}{\eps},
$$
i.e., $\Uel^\eps = \frac{e^2}{2} a_\eps  \geq \e^2/(2 \eps) \to \infty$ as $\eps \to 0$. \qed

As for the \emph{magnetic self-energy} $\Uma$, considered as a generalized number, one may repeat the calculations in \cite[Section VII, Equations (7.1)-(7.14)]{GsponerJMP} with the Heaviside regularizations $H_\eps$ replacing $\Upsilon$. With $\mu$ denoting the norm of the magnetic moment, we have from \cite[Equation (7.14)]{GsponerJMP}
$$
   \Uma^\eps = - \left. \frac{\mu^2}{3 r^3} H_\eps(r)^2 \right|_0^\infty + \frac{\mu^2}{3} \int_0^\infty \frac{H_\eps'(r)^2}{r^2} \, dr.
$$
The boundary term vanishes, since $H_\eps(0) = 0$ and $H_\eps(r)^2/r^2 \to 0$ as $r \to \infty$ for every $\eps \in ]0,1]$. By assumptions (ii) and (iii), we may thus write
$$
  \Uma^\eps = \frac{\mu^2}{3} \int_\eps^{2 \eps} \frac{H_\eps'(r)^2}{r^2} \, dr.
$$
Observe that with $e \neq 0$ and $\mu \neq 0$,
$$
    \frac{2}{e^2} \Uel^\eps =  \int_\eps^{2 \eps} H_\eps'(r)^2 \, dr \leq 
        \int_\eps^{2 \eps} \frac{H_\eps'(r)^2}{r^2} \, dr = \frac{3}{\mu^2} \Uma^\eps 
        \quad (0 < \eps < 1/2),
$$
thus Theorem \ref{UelThm} leads to the following statement.

\begin{corollary} If $\mu \neq 0$, then the generalized number $\Uma$ is infinite in the sense that
$$
   \Uma^\eps \to \infty \quad (\eps \to 0).
$$
\end{corollary}

One could proceed as in \cite{GsponerJMP} with considerations of self-momentum, self-angular momentum, and spin, which would be a straightforward implementation of similar techniques and calculations as above.

\begin{remark}[Total self-energy compared to relativistic mass] We try to illustrate a basic aspect of Gsponer's brief discussion in \cite[Section X]{GsponerJMP} about the opportunity for a ``mass renormalization" in this context. Recall that certain mass renormalizations are also carried out in the classical derivation of the Lorentz-Dirac equation (\cite[Section 8.4]{Thirring:01} and \cite[Section 4.3]{Parrott}).

Note that Gsponer has a second parameter $a > 0$ in his regularizations in \cite{GsponerJMP}, implementing a shift of the Heaviside function to the right by $a$, which we have instead coupled directly to the parameter $\eps$ via the properties in Assumption \ref{AssH}.

When we  compare the total self-energy of the charged point particle 
$$
   \Uel + \Uma =  \frac{\e^2}{2} \int_0^\infty H'(r)^2 \, dr + \frac{\mu^2}{3} \int_0^\infty \frac{H'(r)^2}{r^2} \, dr
$$
to the relativistic energy $m c^2$ for a particle of mass $m$, the status of $\Uel$ and $\Uma$ as generalized numbers in the equation
$$
    m c^2 = \Uel + \Uma
$$
forces us to either consider $m$ also as a generalized number or, in case $m$ refers to a measured classical quantity rather hints potentially at an opportunity of a ``mass renormalization''. Upon choosing a concrete representative $(H_\eps)_{0 < \eps \leq 1}$ and assuming $\e \neq 0$ and $\mu \neq 0$, an attempt at matching this specific regularization of self-energy with a given number $m c^2 > 0$ would require a relation of the form
$$
    m c^2 = \Uel^{\eps} + \Uma^{\eps} =  
    \frac{\e^2}{2} \int_\eps^{2 \eps} H_\eps'(r)^2 \, dr + \frac{\mu^2}{3} \int_\eps^{2 \eps} \frac{H_\eps'(r)^2}{r^2} \, dr.
$$
If, in addition, we have continuity of the map $\eps \mapsto H_\eps$ with respect to (local) Sobolev $H^1$ norms, which can certainly be  achieved by convolution/scaling regularizations, then  the range of the map $\eps \mapsto \Uel^\eps + \Uma^\eps$ is of the form $[b_0,\infty[$. In this case, for any given value of $m c^2$ in this interval, a suitable choice of $0 < \eps_0 \leq 1$ is always possible that will then yield the matching relation
$$
   \Uel^{\eps_0} + \Uma^{\eps_0} = m c^2
$$
for this particular regularization. 

\end{remark}

%%%%%%%%%%%%%%%%%%%%%%%%%%%%%%%%%%%%%%%%%%%%
%%%%%%%. Upsilon-Zeugs
%%%%%%%%%%%%%%%%%%%%%%%%%%%%%%

\section{Justification of $H$ as a model for Gsponer's $\Upsilon$}\label{HUps} 
The original definition of the generalized vector field $\Phi$ in \cite[Equation (4.1)]{GEW} uses a certain generalized function denoted by $\Upsilon$ in place of $H$ in Equation \eqref{Phi}. The function $\Upsilon$ was introduced in \cite{GPOT} and employed in \cite{GEW} and \cite{GEJP} (see also \cite{GEJPCorr}), but lacks a clear definition, at least to our knowledge. In particular, \cite[Theorem in Section 6]{GPOT} characterizes $\Upsilon$ in the form 
\beq\label{UpsEqu}
  \Upsilon (\xi) = \xi \vphi(\xi)
\eeq
via a scalar ``function'' $\vphi$ satisfying the differential equation 
\beq\label{Gsphi}
  \vphi'(\xi) + \frac{1}{\xi} \vphi(\xi) = \frac{1}{\xi} \de_0(\xi). 
\eeq
We will discuss the meaning and solution of this equation in a few variants, which will suggest more or less that $\Upsilon$ is associated with the Heaviside function.

\subsection{Equation \eqref{Gsphi} re-interpreted distributionally:}\label{DistUps}

Avoiding distributional products and having $\xi \vphi(\xi) = \Upsilon(\xi)$ as the main object in mind, we could consider \eqref{Gsphi} formally multiplied by $\xi$ as the simple differential equation for $\vphi\in\mathcal{D}'(\mathbb{R})$ in the form 
\begin{equation}\label{eq2}
    \xi\varphi'(\xi)+\varphi(\xi)=\delta_0(\xi).
\end{equation}
The left-hand side equals $(\xi \vphi(\xi))' = \Upsilon'(\xi)$, hence we conclude immediately that
$$
   \Upsilon (\xi) = \xi  \vphi(\xi) = \theta(\xi) + d
$$
with some constant $d$. The requirement $\Upsilon(\xi) = 1$ for $\xi > 0$ (e.g., in \cite{GEW,GPOT}) now forces $d = 0$, i.e.
\beq\label{U=theta}
   \Upsilon = \theta.
\eeq

\subsection{Re-interpreting Equation \eqref{Gsphi} with generalized functions:} 

We will employ two strategies, one resembling Subsection \ref{DistUps}, the other more directly referring to Equation \eqref{Gsphi}.

Denote by $\mathrm{id}$ the (smooth) identity function $\mathrm{id} \col \xi \mapsto \xi$ on $\R$. 
To begin with, we recall that for any $V \in \G(\R)$ with $V \approx \vp(1/\xi)$ we necessarily have $\mathrm{id} \cdot V \approx \mathrm{id} \cdot \vp(1/\xi) = 1$. But we cannot expect $\mathrm{id} \cdot V = 1$, in fact, there is no generalized function $V$ on $\R$ such that $\mathrm{id} \cdot V = 1$ holds: Suppose $V$ is represented by $(V_\eps)_{0 < \eps \leq 1}$, then for any compact subset $K \subseteq \R$ and any $q > 0$, we could find some $\eps_0 > 0$ such that 
$$
     |\xi V_\eps(\xi) - 1| \leq \eps^q \quad (\xi \in K, 0 < \eps \leq \eps_0),
$$
which is obviously a contradiction, if we take $K := \{ 0\}$.

\subsubsection{By avoiding an interpretation of $1/\xi$}

In view of the above observation, let us re-interpret Equation \eqref{Gsphi} in the context of generalized functions with $\vphi \in \G(\R)$ similarly as in the distribution theoretic setting. As right-hand side, we use a generalized Dirac function $D \approx \de_0$,  represented by a family $(\rho_\eps)_{0 < \eps \leq 1}$ of test functions $\rho_\eps \in \D(\R)$ with $\supp(\rho_\eps)$ contained in a fixed compact set independent of $\eps$, $\int_\R \rho_\eps(\xi) \,d\xi = 1$, and additional requirements to be specified below in \eqref{Remedy}. The corresponding equation then reads
\beq\label{eq3}
        \mathrm{id} \cdot \vphi' + \vphi = D.
\eeq
By essentially the same reasoning as in distribution theory, Equation \eqref{eq3} implies
$$
    \Upsilon' = (\mathrm{id} \cdot \vphi)' = \mathrm{id} \cdot \vphi' + \vphi = D.
$$
Defining
\beq\label{HDef}
  H_\eps(\xi) :=  \int_{-\infty}^\xi \rho_\eps(x) \, d x \quad (\xi \in \R, 0 < \eps \leq 1),
\eeq
the family $(H_\eps)_{0 < \eps \leq 1}$ represents a generalized function $H$ with the property $H' = D$, hence $(\Upsilon - H)' = \Upsilon' - H' = D - D = 0$ and we obtain
$$
    \Upsilon_\eps(\xi) = H_\eps + d_\eps,
$$
where $(d_\eps)_{0 < \eps \leq 1}$ is a moderate family of constants, thus represents a generalized constant $d$ (\cite[Proposition 1.2.35]{GKOS:01}), i.e.,
$$
   \Upsilon = H + d.
$$

We certainly have $H \approx \theta$, since for any test function $f \in \D(\R)$, integration by parts yields
\begin{multline*}
  \dis{H_\eps}{f} = \int_\R H_\eps(\xi) \, f(\xi) \, d\xi
  = \left. \Big(H_\eps(\xi) \int_{-\infty}^\xi f(x) \, dx \Big) \right|_{-\infty}^\infty - \int_\R \rho_\eps(\xi) \int_{-\infty}^\xi f(x) \, dx \, d\xi\\
  = 1 \cdot \!\! \int_{-\infty}^\infty \!\! f(x) \, dx \,- \int_\R \rho_\eps(\xi) \!\! \int_{-\infty}^\xi \!\! f(x) \, dx\, d \xi
  \quad \stackrel{\eps \to 0}{\longrightarrow} \quad   \int_{-\infty}^\infty \!\! f(x) \, dx\;\; - \int_{-\infty}^0 \!\! f(x) \, dx
  = \int_0^\infty f(x)\, dx,
\end{multline*}
where the limit is justified by the support properties of $\supp(\rho_\eps)$ (guaranteeing that $\dis{\rho_\eps}{h} \to h(0)$ holds as $\eps \to 0$  for any $h \in \Cinf(\R)$). 

Employing once again the requirement $\Upsilon(\xi) = 1$ for $\xi > 0$ from \cite{GEW,GPOT},  we may conclude that $d = 0$ and therefore,
\beq\label{UassozTheta}
   \Upsilon = H \approx \theta.
\eeq

It remains to check solvability of \eqref{eq3} via a moderate representing family $(\vphi_\eps)_{0 < \eps \leq 1}$ and consistency of the condition $\mathrm{id} \cdot \vphi = \Upsilon = H$ with the possible choices of the family $(\rho_\eps)_{0 < \eps \leq 1}$ defining $D$ and $H$, since, in the sense generalized numbers, we have to observe
\beq\label{ZeroProblem}
       0 = 0 \cdot \vphi(0) = \Upsilon (0) = H(0).
\eeq
All of this can be achieved by constructing $(\rho_\eps)_{0 < \eps \leq 1}$ in the following way: Choose a nondecreasing function $F \in C^\infty(\R)$ such that $F(\xi) = 0$ for $\xi \leq 1$, $F(\xi) = 1$ for $\xi \geq 2$, and define
\beq\label{Remedy}
   \rho_\eps(\xi) := \diff{\xi} \left( F({\xi}/{\eps}) \right) = F' ({\xi}/{\eps}) / \eps.
\eeq
We have $\supp(\rho_\eps) \subseteq [\eps, 2 \eps]$, $\int_\R \rho_\eps(\xi) \, d\xi = \int_\R F'(z) \, dz = 1$, and also $\rho_\eps \to \de_0$ as $\eps \to 0$, since for any test function $f$,
$$
   \int_\R \rho_\eps(y) f(y) \, dy = \int_\R F'(z) f(\eps z) \, dz \to f(0) \int_\R F'(z) \, dz = f(0). 
$$
Definition \eqref{HDef} gives $H_\eps(\xi) = F(\xi /\eps)$ with $\supp (H_\eps) \subseteq [\eps,\infty[$, so that we may define a smooth function $\vphi_\eps$ by $\vphi_\eps(\xi) := H_\eps(\xi) /\xi$ for $\xi \neq 0$ and $\vphi_\eps(0) := 0$. We obtain 
$$
     \xi \vphi_\eps(\xi) = H_\eps(\xi) \quad\text{and}\quad \xi \vphi_\eps'(\xi) + \vphi_\eps(\xi) = \rho_\eps(\xi).
$$
Finally, moderateness of $(\vphi_\eps)_{0 < \eps \leq 1}$ follows from the expression $\vphi_\eps(\xi) = G(\xi/\eps)/\eps$, where $G \in C^\infty(\R)$ is given by $G(\xi) := F(\xi)/\xi$ ($\xi \neq 0$) and $G(0) := 0$, satisfies $\supp(G) \subseteq [1, \infty[$, and has bounded derivatives. 

\subsubsection{By interpretating $1/\xi$ essentially as $\vp(1/\xi)$}

With the notation $D$ and $H$ as introduced above, let us now experiment with yet another variant of Equation \eqref{Gsphi}. Let $V$ denote a generalized function on $\R$ such that
$$
    V \approx \vp(1/\xi)
$$
and represented by the regularization family $(V_\eps)_{0 < \eps \leq 1}$. Then Equation \eqref{Gsphi} reads
\beq\label{eq4}
     \vphi' + V \vphi = V D
\eeq
and we may supply it with an initial condition $\vphi(1) = c$ with a generalized constant $c$ represented by the moderate family $(c_\eps)_{0 < \eps \leq 1}$ of real numbers. On the level of representatives, the initial value problem reads
$$
   \vphi_\eps' + V_\eps \vphi_\eps = V_\eps \rho_\eps + n_\eps, \quad \vphi_\eps(1) = c_\eps + m_\eps,
$$
where $(n_\eps)$ is a negligible family of smooth functions and $(m_\eps)$ is a negligible family of real numbers. At fixed and arbitrary $0 < \eps \leq 1$, the unique solution is given by the smooth function
\beq\label{SolRep}
   \xi \mapsto  (c_\eps + m_\eps) e^{- \int_1^\xi V_\eps(y) \, dy} + 
     \int_1^\xi e^{\int_\xi^z V_\eps(y) \,dy} (V_\eps(z) \rho_\eps(z) + n_\eps(z))\, dz.
\eeq
This defines a moderate family of smooth functions under the condition that for every compact interval $I \subseteq \R$, 
\beq\label{IntBo}
   \sup_{a, b \in I} \left| \int_a^b V_\eps(y) \, dy \right| = O\left(\log\left(\frac{1}{\eps}\right)\right)
\eeq
holds. In this case, the terms involving $n_\eps$ and $m_\eps$ in the above formula are negligible and we may therefore define a representative $(\vphi_\eps)$ of a solution $\vphi \in \G(\R)$ to Equation \eqref{eq4} by
\beq\label{phiRep}
   \vphi_\eps(\xi) :=  c_\eps e^{- \int_1^\xi V_\eps(y) \, dy} +  \int_1^\xi e^{\int_\xi^z V_\eps(y) \,dy} V_\eps(z) \rho_\eps(z)\, dz.
\eeq
We obtain uniqueness of the solution in $\G(\R)$ under the same condition \eqref{IntBo}, since any solution of the homogenous equation $\vphi' + V \vphi = 0$ in $\G(\R)$ with negligible initial value necessarily has a negligible family of solution representatives by the formula in \eqref{SolRep} (this is similar to the more general situation described in \cite[Theorem 1.5.2]{GKOS:01}). The asymptotic condition reminds also of a H\"older-Zygmund regularity property and techniques from \cites{Hoermann:04} could be employed to analyze qualitative properties of the solution.

\begin{example}\label{VasVP}
Embedding $\vp(1/\xi)$ into $\G(\R)$ as generalized function $V$ via a model delta net $\mu_\eps(y) := \mu(y/\eps)/\eps$ with $\mu \in \D(\R)$, $\int \mu(y) \, dy = 1$, by defining the representative $(V_\eps)_{0 < \eps \leq 1}$ in the form $V_\eps := \vp(1/\xi) * \mu_\eps$, provides us with an example satisfying \eqref{IntBo}. Indeed, we have
\begin{multline*}
   \left| \int_a^b V_\eps(y) \, dy \right| = \left| \int_a^b \int_\R (\log |z|) \mu_\eps'(y - z) \, dz \, dy \right|
   =  \left| \int_a^b \diff{y} \left( \int_\R (\log |z|) \mu_\eps(y - z) \, dz \right) dy \right|\\
   = \left| \int_\R (\log |y|) \mu_\eps(b - y) \, dy -  \int_\R (\log |y|) \mu_\eps(a - y) \, dy \right| \\
   =  \left| \int_\R (\log |b - y|) \mu_\eps(y) \, dy -  \int_\R (\log |a - y|) \mu_\eps(y) \, dy \right|.
\end{multline*} 
In case $a \neq 0$ and $b \neq 0$ we obtain $| \log |b| - \log |a| |$ in the limit $\eps \to 0$, hence the integral is bounded independently of $\eps$. The same is trivially true in case $a=0$ and $b = 0$. Finally, in case $a = 0$ and $b \neq 0$ (which is analogous to the case $a \neq 0$ and $b =0$) we have for some $\eps_0 > 0$ for all $0 < \eps < \eps_0$ the estimate
\begin{multline*}
     \left| \int_0^b V_\eps(y) \, dy \right| = 
     \left| \int_\R (\log |b - y|) \mu_\eps(y) \, dy -  \int_\R (\log |y|) \mu_\eps(y) \, dy \right| \\
     \leq  \left| \int_\R (\log |b - y|) \mu_\eps(y) \, dy \right|  + \left| \int_\R (\log |y|) \mu_\eps(y) \, dy \right| 
     \leq 2 | \log |b| | + \int_\R | \log |y| | \, |\mu_\eps(y)| \, dy,
\end{multline*}
since $ \int_\R (\log |b - y|) \mu_\eps(y) \, dy \to \log |b|$ as $\eps \to 0$. In the last term we use the change of variables $y = \eps z$ and obtain
\begin{multline*}
    \int_\R | \log |y| | \, |\mu_\eps(y)| \, dy = \int_\R | \log |\eps z | | \, |\mu(z)| \, dz = 
        \int_\R | \log | z | + \log \eps | \, |\mu(z)| \, dz  \\
        \leq  \int_\R | \log |z | | \, |\mu(z)| \, dz +  \int_\R | \log \eps | \, |\mu(z)| \, dz
        =  \int_\R | \log |z | | \, |\mu(z)| \, dz + \log\left(\frac{1}{\eps}\right)  \int_\R  |\mu(z)| \, dz.
\end{multline*}
To summarize, we may conclude that we have $\sup_{a,b \in I} \left| \int_a^b V_\eps(y) \, dy \right| = O(\log(1/\eps))$ as $\eps \to 0$ for any compact interval $I \subseteq \R$.
\end{example}

It remains to analyze the solution $\vphi$, given by the representative according to \eqref{phiRep} and to determine $\Upsilon = \mathrm{id} \cdot \vphi$ in a situation where $V$ can be considered an appropriate model of $1/\xi$. The minimal requirement  $V \approx \vp(1/\xi)$ without further specification seems too vague to obtain a distributional limit for the regularized solutions according to \eqref{phiRep}. Instead of pursuing a general search for suitable properties replacing the impossible condition $\mathrm{id} \cdot V = 1$, we will be content with a little further analysis of the special case considered in Example \ref{VasVP}, where $V$ is an embedding of the Cauchy principal value.

\begin{example}[Example \ref{VasVP} revisited---regularization dependence of the dsitributional limit] We weaken Gsponer's requirement $\Upsilon(\xi) = 1$ for $\xi > 0$ from \cite{GEW,GPOT} and replace it by the initial value condition $1 = \Upsilon(1) = 1 \cdot \vphi(1) = c$ (in the sense of generalized numbers), hence we may simplify matters by assuming $c_\eps = 1$ in \eqref{phiRep}. We then have
$$
   \psi_\eps := \mathrm{id} \cdot \vphi_\eps = S_\eps + T_\eps \quad\text{with}\quad 
       S_\eps(x) := x e^{- \int_1^x V_\eps(y) \, dy} \quad\text{and}\quad
       T_\eps(x) := x \int_1^x e^{\int_x^z V_\eps(y) \,dy} V_\eps(z) \rho_\eps(z)\, dz
$$
and recall  $V_\eps := \vp(1/\xi) * \mu_\eps$ with the model delta net $\mu_\eps$ as in Example \ref{VasVP}. Let us now look at the specific case, where $\mu$ is nonnegative and also $\rho_\eps(x) := \rho(x/\eps)/\eps$ is a model delta net with $\rho \in \D(\R)$ such that $\int_\R \rho(x) \, dx = 1$, and study the convergence properties of $\psi_\eps$ as $\eps \to 0$  in the sense of $\D'(\R)$ for this particular configuration in terms of $S_\eps$ and $T_\eps$ separately.

We claim that
\beq\label{Claima}
   S_\eps \to \sgn. 
\eeq

Let $f \in \D(\R)$ and consider
\beq\label{Saction}
    \dis{S_\eps}{f} = \int_\R f(x) x e^{- \int_1^x V_\eps(y) \, dy} \,dx.
\eeq
We have $S_\eps(0) = 0$ and learn from Example \ref{VasVP} that we have for every $x \neq 0$ pointwise convergence  
$$
    x e^{- \int_1^x V_\eps(y) \, dy} \to x e^{- \log|x|} = \frac{x}{|x|} = \sgn(x) \quad (\eps \to 0).
$$
To obtain an upper bound for the integrand, we note that $V_\eps = \diff{x}(\log |.| * \mu_\eps)$ and hence 
\begin{multline*}
   \int_1^x V_\eps(y) \, dy = (\log |.| * \mu_\eps)(x) - (\log |.| * \mu_\eps)(1)\\ = 
   \int_\R \big(\log|x - r| - \log|1-r|\big) \mu\Big(\frac{r}{\eps}\Big) \,\frac{dr}{\eps}
   = \int_\R  \big(\log|x - \eps s| - \log|1- \eps s|\big) \mu(s) \,ds\\
   =  \int_\R  \Big(\log \Big|\frac{x}{\eps} - s\Big| - \log \Big|\frac{1}{\eps}- s \Big|\Big) \mu(s) \,ds 
   = L\Big(\frac{x}{\eps}\Big) - L\Big(\frac{1}{\eps}\Big),
\end{multline*}
where we have put $L(y) := (\log |.| * \mu)(y)$. Choose $\be > 0$ such that $\supp(\mu) \subseteq [-\be,\be]$. 
We observe in case $y \neq 0$ that $\log |y-z| = \log|y| + \log |1 - z/y|$ and $\int_\R \mu(z) \, dz = 1$ yield 
$$
   L(y) = \int_\R \mu(z) \log|y-z| \, dz = \log|y| + \int_{-\be}^\be \log|1 - z/y| \mu(z) \, dz.
$$
If $|y| > \be$ then the absolute value of the last term  is bounded by
$$
   \int_{-\be}^\be \Big| \log|1 - z/y| \Big| |\mu(z)| \, dz \leq \|\mu\|_{L^1} \cdot 
       \sup_{|z| \leq \be} |\log |1 - z/y| |,
$$
thus
$$
    L(y) = \log|y| + O(1/|y|) \quad (\be < |y| \to \infty).
$$
This gives $\eps e^{L(1/\eps)} \to 1$ as $1/\be > \eps \to 0$ and $\frac{|x|}{\eps} e^{- L(x/\eps)} = \frac{|x|}{\eps} \frac{\eps}{|x|} e^{O(\eps/|x|)} = O(1)$ as $|x|/\be > \eps \to 0$, which implies an upper bound independent of $\eps$ for the absolute value of
$$
    S_\eps(x) = x e^{- \int_1^x V_\eps(y) \, dy} = x e^{L(1/\eps)- L(x/\eps)}
    = \eps e^{L(1/\eps)} \cdot \frac{x}{\eps} e^{- L(x/\eps)}
$$
in case $|x| \geq \be \eps$; in case $|x| < \be \eps$, an inspection of the arguments in Example \ref{VasVP} (using $\int |\mu| = \int \mu$) give an estimate $\left| \int_1^x V_\eps(y) \, dy \right| \leq c + \log(1/\eps)$ with a constant $c \geq 0$ independent of $\eps$. Therefore $|S_\eps(x)| \leq \be \eps \, e^{c + \log(1/\eps)} = \be e^c \eps \cdot (1/ \eps) = O(1)$ as $\eps \to 0$. In summary, we conclude that the integrand in \eqref{Saction} is bounded by a constant (independent of $\eps$) times $|f(x)|$, thus dominated convergence applies and yields $ \dis{S_\eps}{f} \to  \dis{\sgn}{f}$ as $\eps \to 0$. Thus we have proved the claim in \eqref{Claima}.

We further claim that, with the notation introduced above,
\beq\label{Claimb}
    T_\eps \to C(\mu,\rho) \,\check{\theta}, 
   \text{ where $\check{\theta}(x) := \theta(-x)$ and $C(\mu,\rho) := \int_\R e^{L(s)} V_1(s) \rho(s)\, ds$}.
\eeq
(Recall that both $L$ and $V_1$ involve $\mu$ in their defining integrals, hence the notation $C(\mu,\rho)$.)

We start out by noting that $V_\eps(z) = V_1(z/\eps)/\eps$ and, as above, we also have $\int_x^z V_\eps(y) \,dy = L(z/\eps) - L(x/\eps)$. A change of variables $z = \eps s$ thus gives
$$
   T_\eps(x) = x e^{- L(x/\eps)} \int_1^x e^{L(z/\eps)} \frac{V_1(z/\eps)}{\eps} \frac{\rho(z/\eps)}{\eps}\, dz
   = \frac{x}{\eps} e^{- L(x/\eps)}  \int_{1/\eps}^{x/\eps} e^{L(s)} V_1(s) \rho(s) \, ds.
$$
If $x > 0$, then the last integral will miss the (compact) support of $\rho$ as soon as $\eps$ is sufficiently small, hence we have pointwise convergence $T_\eps(x) \to 0$ on the positive half line. If $x < 0$, then we have seen earlier that $ \frac{x}{\eps} e^{- L(x/\eps)} \to -1$ as $\eps \to 0$; in addition, the interval $[x/\eps,1/\eps]$ will contain the entire support of $\rho$ as soon as $\eps$ is sufficiently small, hence we have pointwise convergence  (note the reversed integration limits on the way)
$$
    T_\eps(x) \to  \int_\R e^{L(s)} V_1(s) \rho(s) \, ds = C(\mu,\rho) \quad (x < 0, \eps \to 0).
$$
Boundedness (uniformly with respect to $\eps$ and $x$) of the integral factor in the above expression for $T_\eps$ is obvious and is deduced for the other factor from the boundedness of $S_\eps$. Therefore, Lebesgue's dominated convergence theorem implies $\dis{T_\eps}{f} \to  \dis{C(\mu,\rho) \check{\theta}}{f}$ as $\eps \to 0$ for every $f \in \D(\R)$, which proves claim \eqref{Claimb}.

In summary, we obtain the distributional convergence
$$
 \psi_\eps = \mathrm{id} \cdot \vphi_\eps = S_\eps + T_\eps \to \sgn + C(\mu,\rho) \, \check{\theta} \quad (\eps \to 0),
$$
or in terms of association,
\beq\label{UpsRegAss}
   \Upsilon = \mathrm{id} \cdot \vphi \approx \sgn + C(\mu,\rho) \, \check{\theta},
\eeq
where the constant $C(\mu,\rho)$ depends on the concrete model delta net regularizations defining $\vp(1/\xi)$ and $D$.

For example, if we choose both $\mu$ and $\rho$ to be even functions, then $L$ is even and $e^L V_1$ is odd yielding $C(\mu,\rho) = 0$, which results in 
$$
  \Upsilon \approx \sgn.
$$

Only a choice of $\mu$ and $\rho$ such that $C(\mu,\rho) = 1$ could achieve a match with \eqref{UassozTheta}, namely
\beq\label{UassozTheta2}
   \Upsilon = \mathrm{id} \cdot \vphi \approx \sgn + \, \check{\theta} = \theta.
\eeq

\emph{It remains to make sure that a choice of $\mu$ and $\rho$ with $C(\mu,\rho) = 1$ is indeed possible:} Take an arbitrary $\mu \in \D(\R)$ with $\int \mu = 1$. With the notation introduced above, we have $U := (e^L)' = e^L \cdot V_1$, which we may consider as a distribution on $\R$ acting on a test function $f \in \D(\R)$ by 
$$
    \dis{U}{f} = \int_\R e^{L(s)} V_1(s) f(s) \, ds = - \int e^{L(s)} f'(s) \, ds.
$$
Note that $C(\mu,\rho) = \dis{U}{\rho}$, if $\rho \in \D(\R)$ with $\int \rho =1$.
The affine hyperplane of these mollifiers  can be written as $\ker(1) + \{ \rho_0\}$, where $\rho_0 \in \D(\R)$ is any function with $\int \rho _0 = 1$. Since $U$ is not the constant function, we cannot have $\ker(1) \subseteq \ker U$, hence there exists some $\chi_0 \in \ker(1)$ with $\dis{U}{\chi_0} \neq 0$. Now, $\chi_0 \in \ker(1)$ means $\chi_0 \in \D(\R)$ and $\int \chi_0 = 0$, so that 
$$
   \rho := \rho_0 + \frac{1 - \dis{U}{\rho_0}}{\dis{U}{\chi_0}} \chi_0
$$
defines a test function $\rho \in \D(\R)$ with $\int \rho =  \int \rho_0 + \frac{1 - \dis{U}{\rho_0}}{\dis{U}{\chi_0}} \int \chi_0 = 1$. Finally, we also have
$$
   C(\mu,\rho) = \dis{U}{\rho} = \dis{U}{\rho_0} +  \frac{1 - \dis{U}{\rho_0}}{\dis{U}{\chi_0}} \dis{U}{\chi_0} = 1.
$$

\end{example}

%\vspace{1cm}
%
%\todo{... obsolet ... }
%multiply both sides by $\mathrm{id}$ and calculate for the generalized function $\Upsilon := \mathrm{id} \cdot \vphi$, within the algebra of generalized functions, \todo{argument includes $1 = \mathrm{id} \cdot V$, which cannot be true; only $1 \approx \mathrm{id} \cdot V$ is possible ...}
%$$
%    \Upsilon' = (\mathrm{id} \cdot \vphi)' = \mathrm{id} \cdot \vphi' + \vphi = \mathrm{id} \cdot \vphi' +  (\mathrm{id} \cdot V) \cdot \vphi
%    = \mathrm{id} \cdot (\vphi' + V \vphi) = \mathrm{id} \cdot (V D) = \mathrm{id} \cdot V \cdot D = D.  
%$$
%We have $H' = D$, hence $(\Upsilon - H)' = \Upsilon' - H' = D - D = 0$ and we obtain
%$$
%    \Upsilon_\eps(\xi) = H_\eps + d_\eps,
%$$
%where $(d_\eps)_{0 < \eps \leq 1}$ is a moderate family of constants, thus represents a generalized constant $d$, i.e.,
%$$
%   \Upsilon = H + d.
%$$
%
%\todo{Bis hierher wahrscheinlich obsolet ...}

\subsection{Remarks on attempting to interpret Equation \eqref{Gsphi} with distributional products:}\label{ODE_DP}\label{IntUps}

When interpreting the factors $1/\xi$ as Cauchy principal value $\vp(\frac{1}{\xi})$, then the right-hand side in \eqref{Gsphi} could be considered as distributional product $\vp(\frac{1}{\xi}) \cdot \delta_0(\xi)$ in the context of strict regularization products (cf.\ \cite[\pg 7]{O:92}) and gives the value $- \de_0'/2$ (see \cite[Equation (7.7)]{O:92}). Therefore, Equation \eqref{Gsphi}  would then read
\beq\label{DPEqu}
     \vphi'(\xi) + \vp\Big(\frac{1}{\xi}\Big) \cdot \vphi(\xi) = - \frac{1}{2} \de_0'(\xi),
\eeq
which still leaves the meaning of the product $\vp(1/\xi)  \cdot \vphi(\xi)$ undefined. We have to point out that non-associativity (\cite[Remark after Corollary 7.3]{O:92}) of the strict product that we have employed above (of type (7.4) as it is labeled in \cite{O:92}) now rules out to make use of the distributional result $\xi \cdot \vp(1/\xi) = 1$ as a standard multiplication by a $\Cinf$ function. However, supposing that there is some distribution satisfying \eqref{DPEqu}, then this implies for its restriction $\vphi_0$ to $\R\setminus \{0\}$ the equation $\xi \vphi_0'(\xi) + \vphi_0(\xi) = 0$ if $\pm \xi > 0$. Hence there are constants $c_-$ and $c_+$ such that
$$
    \vphi_0(\xi) = \frac{c_-}{\xi} \;\; (\xi < 0) \quad\text{and}\quad \vphi_0(\xi) = \frac{c_+}{\xi} \:\:(\xi > 0),
$$
thus $\vphi_0$ is homogeneous of degree $-1$ on $\R\setminus \{0\}$. By \cite[Theorem 3.2.4]{Hoermander:V1} (noting that we have $a=-1 = -n$ and $k=0$ in our case), $\vphi_0$ possesses an extension to $\vphi \in \D(\R)$ with a weakened, logarithmically perturbed, homogeneity property (as in \cite[Equations (3.2.24) or (3.2.24)$'$]{Hoermander:V1}), which is determined up to multiple of $\de_0$. Employing the definition of the distributions $\xi_\pm^{-k}$ ($k \in \N$) as in \cite[Equation (3.2.5)]{Hoermander:V1}, we may write
$$
    \vphi(\xi) = c_- \xi_-^{-1} + c_+ \xi_+^{-1} + c_0 \de_0(\xi)
$$
with some constant $c_0$ and conclude that $\Upsilon(\xi) = \xi \vphi(\xi) = - c_- \theta(-\xi) + c_+ \theta(\xi) +  0 = c_+ \theta(\xi) - c_- \theta(-\xi)$.
Requiring  $\Upsilon(\xi) = 1$ for $\xi > 0$ (as is done in \cite{GEW,GPOT}) leads to the intermediate result
$$
    \Upsilon(\xi) = \theta(\xi) - c_- \theta(-\xi)
$$
and leaves the values on the negative real line undetermined. 

Any attempt in checking Equation \eqref{DPEqu} would need to evaluate also the second term on the left-hand side, i.e., would require to make sense of the distributional product
\beq\label{DisPro}
      \vp\Big(\frac{1}{\xi}\Big) \cdot \vphi(\xi) =  \vp\Big(\frac{1}{\xi}\Big) \cdot \big( c_- \xi_-^{-1} + c_+ \xi_+^{-1} + c_0 \de_0(\xi) \big),
\eeq
besides applying \cite[Equation (3.2.2)${}''$, page 69]{Hoermander:V1}, which gives $(\xi_\pm^{-1})' = \mp \xi_\pm^{-2} - \de_0'$, and hence determines the derivative in the first term on the left-hand side of Equation \eqref{DPEqu} as
$$
    \vphi'(\xi) =  c_- \xi_-^{-2} - c_- \de_0'  - c_+ \xi_+^{-2}  - c_+ \de_0' + c_0 \de_0'(\xi) = 
            c_- \xi_-^{-2} - c_+ \xi_+^{-2} + (c_0 - c_- - c_+) \de_0'.
$$

If the product \eqref{DisPro} were possible within the hierarchy of distributional products in \cite[\pg 7, page 69]{O:92}, then, knowing already that $\vp(1/\xi) \cdot \de_0 = -\de_0'/2$ exists, we would deduce that 
$$
     \vp\Big(\frac{1}{\xi}\Big) \cdot \big( c_- \xi_-^{-1} + c_+ \xi_+^{-1} \big) 
$$
had to exist at least as a model product (of type (7.4) in \cite{O:92}).  Recalling (from \cite[(3.2.10)${}'$ and (3.2.14)]{Hoermander:V1}) 
$$
    \vp\Big(\frac{1}{\xi}\Big) = \xi_+^{-1} - \xi_-^{-1},
$$ 
one is thus tempted to look into the question of existence and meaning of the products $\xi_+^{-1} \cdot \xi_+^{-1}$, $\xi_+^{-1} \cdot \xi_-^{-1}$, and $\xi_-^{-1} \cdot \xi_-^{-1}$: We learn from \cite[Theorem 2]{Itano} that the product $x_+^{-1} \cdot x_+^{-1}$ does not exist as a limit of analytic regularizations, hence it follows from \cite[Remark 8.4]{O:92} that $x_+^{-1} \cdot x_+^{-1}$ cannot exist in the context of the most general (model) product within the hierarchy in \cite[page 69]{O:92}.

We may mention that, according to \cite[Theorem 3]{Itano}, the product $\xi_+^{-1} \cdot \xi_+^{-1}$ exists in the sense of taking finite parts of non-convergent analytic regularizations, but we were not able to arrive at a conclusion like $\Upsilon = \theta$, i.e., $c_- = 0$, based on such products.

\bigskip 

\subsection*{Acknowledgement:} We are thankful to an anonymous referee for pointing out a number of flaws and inconsistencies in the previous version of this manuscript. All improvements in this revised version are due to the careful remarks and suggestions of the referee.

\bibliography{GN}
\bibliographystyle{abbrv}

\end{document}